\documentclass[11pt,english]{article}
\usepackage[T1]{fontenc}
\usepackage[latin9]{inputenc}
\usepackage{verbatim}
\usepackage{amsmath}
\usepackage{graphicx}
\usepackage{amssymb}
\usepackage{esint}

\makeatletter

\usepackage{amsthm}
\usepackage{cite}\@ifundefined{definecolor}
 {\@ifundefined{definecolor}
 {\usepackage{color}}{}
}{}
\usepackage{pdfsync}

\theoremstyle{plain}
\newtheorem{thm}{Theorem}
\theoremstyle{plain}
\newtheorem{corl}[thm]{Corollary}
\theoremstyle{plain}
\newtheorem{lem}[thm]{Lemma}
\theoremstyle{plain}
\newtheorem{prop}[thm]{Proposition}

\newcommand{\s}{\widehat{s}}

\global\long\def\<{\langle}
 \global\long\def\>{\rangle}
 \global\long\def\R{\mathbb{R}}
 \global\long\def\C{\mathbb{C}}
 
 \global\long\def\E{\mathbb{E}}
 \global\long\def\P#1{\operatorname{P}\left\{  #1\right\}  }
  \global\long\def\PE#1{\operatorname{P}_{\epsilon_k}\left\{  #1\right\}  }
  \global\long\def\PN#1{\operatorname{P}_{n_k}\left\{  #1\right\}  }
 \global\long\def\iunit{\mathrm{i}}
\global\long\def\bigmone{M_{1}}
\global\long\def\bigmtwo{M_{2}}
\global\long\def\bigmthree{M_{3}}
\global\long\def\bigmfour{M_{4}}
\global\long\def\TO{|\Omega||T|}

\makeatother
\usepackage{babel}

\begin{document}

\title{Matched Filtering from Limited Frequency Samples}

\author{Armin Eftekhari,$^{c}$ Justin Romberg,$^{g}$ and Michael B. Wakin$^{c}$%
\thanks{
Email: aeftekha@mines.edu, jrom@ece.gatech.edu, mwakin@mines.edu.
This work was partially supported by NSF Grants CCF-0830320 and CNS-0910592 and by DARPA grant FA8650-08-C-7853.
Submitted to the {\em IEEE Transactions on Information Theory} on January 13, 2011. %
} \\
 \vspace{-0.1cm}
 \\
 $^{c}$ Dept.\ of Electrical Engineering and Computer Science, Colorado School of Mines\\
 {} \vspace{-0.5cm}
 \\
 {} $^{g}$ School of Electrical and Computer Engineering, Georgia
Tech}

\date{January 2011; revised July 2012}

\maketitle

\begin{abstract}
 In this paper, we study a simple correlation-based strategy for estimating
 the unknown delay and amplitude of a signal based on a small number
 of noisy, randomly chosen frequency-domain samples. We model the output
 of this ``compressive matched filter'' as a random process whose
 mean equals the scaled, shifted autocorrelation function of the template
 signal. Using tools from the theory of empirical processes, we prove
 that the expected maximum deviation of this process from its mean
 decreases sharply as the number of measurements increases, and we
 also derive a probabilistic tail bound on the maximum deviation. Putting
 all of this together, we bound the minimum number of measurements
 required to guarantee that the empirical maximum of this random process
 occurs sufficiently close to the true peak of its mean function. We
 conclude that for broad classes of signals, this compressive matched
 filter will successfully estimate the unknown delay (with high probability,
 and within a prescribed tolerance) using a number of random frequency-domain
 samples that scales inversely with the signal-to-noise ratio and only
 logarithmically in the observation bandwidth and the possible
 range of delays.
\end{abstract}

\section{Introduction}

\subsection{Random Sampling and Compressive Signal Processing}

Over the last few decades, the development of cheap, flexible, and
powerful digital signal processing (DSP) architectures has enabled
the acquisition and analysis of increasingly rich data sets. One of
the key principles behind the DSP revolution is the fundamental work
by Nyquist, Whittaker, and Shannon in characterizing the minimum number
of discrete-time samples required to fully capture the information
in a bandlimited continuous-time signal. Unfortunately, many real-world
signals of interest may have very high bandwidth, which can severely
complicate the practical task of sampling a signal at its Nyquist
rate~\cite{le2005adcs,tropp2010beyond}.

The recently developed theory of Compressive Sensing (CS)~\cite{candes2006robust,donoho2006compressed}
suggests that if a signal is structured, then we can acquire it by
taking samples well below its Nyquist rate. CS relies on two fundamental
principles: first, that many signals have much lower complexity than
is suggested by their bandwidth (typically this is embodied in a sparse
representation for the signal within some basis), and second, that
such signals may safely be sampled below their Nyquist rate if the
traditional uniform time-domain sampling procedure is replaced with
a generalized linear measurement operator (typically this operator
contains some degree of randomness).

The CS theory has benefited from several powerful and elegant tools
for probabilistic analysis relating to the theory of empirical processes.
The essential condition (the {\em restricted isometry property}
\cite{candes2006nearoptimal}) that guarantees sparse recovery from
observations through a random matrix can be recast as a bound on a
random process --- this formulation, first put forth in \cite{Rudelson-Dudley},
is particularly useful when the compressive measurement system is
structured \cite{tropp2010beyond,ra09,rombergGaussian2009,rombergNeelsh2010,rauhut10re}.
In these works, the Dudley inequality \cite{dudley1967sizes}, a classical
tool which relates the supremum of a random process to the geometry
of its index set, is used to bound the expected supremum of the process,
and strong tail bounds are established that control the deviation
from the average behavior. To date, almost all of the work along these
lines has focused on providing guarantees for signal recovery from
compressive measurements.

There are many applications, however, where we are not interested in a full-scale
recovery of a signal. Instead, we may wish only to estimate some key
parameters (or {}``features'') in order to solve an inference problem that
does not demand full knowledge of the signal. It has been demonstrated
that random measurements can again be very useful in such settings.
Just as certain low-complexity signals can be fully recovered from
random measurements, certain low-complexity questions can be answered
about (possibly arbitrary) signals directly from random measurements
without first recovering the signal. Some initial steps in this direction
have been concerned with compressive detection, classification, estimation,
and filtering~\cite{haupt2006compressive,davenport2006sparse,HauptDetect,davenport2007smashed,davenport2009signal}.
Compared to alternative techniques that base their inference on a
full set of Nyquist samples, compressive inference techniques can
show slightly diminished accuracy because fewer statistics are measured
concerning the signal. In exchange, the acquisition hardware can potentially be much simpler and consume less power.  In addition,
we maintain the ability to adapt to future information we may learn the
problem at hand; from a single set of random measurements, a number
of different inference problems may be solved concerning a number
of possible candidate signals.

In this paper, we study the problem of matched filtering (i.e., estimating
the unknown delay and amplitude of a known template signal) from the
compressive viewpoint. In particular, we derive strong bounds on the
performance of a compressive matched filter by bringing in some of
the same probabilistic tools that have been so fruitful in CS recovery
analysis. To do this, we show that the compressive matched filtering problem can be reduced to controlling the supremum of a certain random process, which we approach through a specialized version of the Dudley inequality.

\subsection{Matched Filtering from Limited Frequency Samples}

\subsubsection{The Compressive Matched Filter}

The problem that we consider is formally stated as follows. Suppose we
make observations of the continuous-time signal $A\cdot s_0(t-\tau_{0})$,
where $s_0(t)$ is a known signal template, $\tau_{0}\in T\subset\R$
is the unknown delay (the {\em time-of-arrival}), and $A \in \R$ or $\C$
is the unknown amplitude. Given these observations, we would like
to estimate $\tau_{0}$ and $A$.

The optimal solution to this problem is given by the {\em matched filter}.  All shifts of the known template $s_0(t)$ are correlated against the incoming signal, and the estimated time-of-arrival is the shift that yields the maximum correlation.  The matched filter is typically implemented in one of two ways: either with a specialized analog circuit that performs the correlation and then detects the peaks, or by sampling the signal and calculating the correlation function digitally.  The advantage of the digital approach is the flexibility it offers; we can perform matched filtering against many different waveforms from the same set of samples, in case $s_0(t)$ is not known in advance (but belongs to a collection of candidates).  If the signal $s_0(t)$ is concentrated in time, the sampling rate must be commensurately high to accurately estimate $\tau_0$.  Applications such as high-frequency radar or ultra-wideband communications can require sampling rates of hundred of millions, or even billions, of samples per second.  Taking and processing samples at these kinds of rates is costly in terms of hardware complexity and power consumption.

Working from compressed samples gives us a more elegant solution. In this paper we analyze a simple correlation-based estimator that operates using a small number $m$ of noisy samples of spectrum of the received signal, with locations that are drawn randomly from a uniform distribution on some interval $\Omega$ in the frequency domain. In one of our main results (Corollary~\ref{corl:noise}), we prove that
for broad classes of signals, this compressive matched filter will successfully estimate $\tau_{0}$ (with high probability, and within
a prescribed tolerance) using a number of random frequency-domain samples that scales inversely with the signal-to-noise ratio and only logarithmically in the observation bandwidth $|\Omega|$ and the possible range of delays $|T|$. Our results help validate the use of compressive measurements for capturing important signal information. This acquisition scheme also offers us flexibility in that it depends only on very broad characteristics of the signal; it is universally effective for all $s_0(t)$ which are spread out over the band $\Omega$.

We note that the use of randomized measurements in the frequency domain is not unprecedented. One of the original motivating problems for CS, for example, came from magnetic resonance imaging (MRI), where the goal is to reconstruct an image from a partial set of Fourier coefficients~\cite{candes2006robust,lustig2008compressed}. Randomized frequency-domain measurements are also standard in CS problems where the signal to be recovered is sparse in the time domain~\cite{candes2006robust,Rudelson-Dudley}, and of course, in the compressive matched filter problem, the unknown signal delay is manifested in the time domain. In hardware, the requisite spectral samples for the compressive matched filter could be acquired by correlating the incoming signal with a bank of oscillators tuned to random frequencies, or by following a Fourier transforming device (such as a SAW processor~\cite{jack1980theory}) with a random sampler. Although the analysis in our paper is limited to one-dimensional signals, one could also envision formulating the matched filtering problem for two-dimensional images, and random samples of a two-dimensional spectrum could be acquired by combining the Fourier transforming property of a lens with a random sampler.

\subsubsection{Analytical Framework and Summary of Main Results}

In this paper, we develop an analytical framework for studying the compressive matched filter based on tools from probability theory and empirical processes.
To help build intuition, in Sections~\ref{sec:problem} through~\ref{sec:gnpc} we first study the problem fully in the case of noiseless measurements.
In Section~\ref{sec:Robustness} we then extend all of our analysis in parallel fashion to account for measurement noise.
Section~\ref{sec:theory} and several appendices provide supporting proofs for all of our main results.

For both the noiseless and noisy problem formulations, we begin by showing that the output of the correlation-based estimator can be modeled as a random process whose mean equals the scaled, shifted autocorrelation function of the template signal.
Noting that the scaled, shifted autocorrelation function of the template signal (the mean of this random process) peaks at $\tau_0$, we estimate $\tau_0$ by finding the empirical maximum of the random process, and we give guarantees about the accuracy of this estimate by showing that the random process does not vary too much from its mean.
Given the estimate of the delay, an estimate of the amplitude $A$ follows easily via least-squares, just as with the standard matched filer.

We approach the analysis as follows.
In Theorems~\ref{th:Emain} and~\ref{thm:noiseexp}, we adapt the proof of the Dudley inequality to show that the expected maximum deviation of this random process from its mean decreases sharply as the number of measurements increases. A bit more formally, Theorem~\ref{th:Emain} states that in the noiseless case, the expected maximum deviation of this process from its mean decreases roughly like $m^{-1/2}$ (normalized by the peak value of the mean function). Theorem~\ref{thm:noiseexp} quantifies the amount of additional deviation one would expect based on noise in the observations.
In Theorems~\ref{th:Pmain} and~\ref{thm:mainNoisenew}, we then derive a probabilistic tail bound on the maximum deviation of this process from its mean. Specifically, Theorem~\ref{th:Pmain} guarantees that with high probability the noiseless process stays uniformly close to its mean, and Theorem~\ref{thm:mainNoisenew} guarantees that with high probability the maximum additional deviation caused by noise is also bounded.
Finally, in Corollaries~\ref{corl:main} and~\ref{corl:noise}, we pull these results together to establish bounds on the number of measurements
required to guarantee that the empirical maximum of this random process occurs sufficiently close to the true peak of its mean function. Specifically, Corollary~\ref{corl:main} ensures in the noiseless case that when the template signal has an autocorrelation function with a single prominent peak, no values of $\tau$ far from $\tau_0$ can yield an estimate close to the true peak. Corollary~\ref{corl:noise} extends this result to account for noise and leads to the central result: the compressive matched filter will successfully estimate $\tau_{0}$ from $m$ random frequency-domain samples (with high probability, and within a prescribed tolerance) as long as $m$ scales inversely with the signal-to-noise ratio and logarithmically in the observation bandwidth $|\Omega|$ and the possible range of delays $|T|$.

All of our bounds depend on the degree to which the template signal $s_{0}$ is concentrated in the frequency domain.
As might be expected given the uniform random sampling strategy on $\Omega$, signals whose spectrum is relatively flat across $\Omega$ will require the fewest measurements, while signals with highly peaked spectra will require the most.
These issues are carefully quantified and discussed in detail throughout Section~\ref{sec:Matched filter noise free}.

\subsubsection{Exchanging Time and Frequency}
\label{sec:exchtf}

It is important to point out that the roles of time and frequency are completely interchangeable in our analysis. All of our results from Section~\ref{sec:Matched filter noise free} can therefore be adapted to the ``dual'' problem of estimating the unknown carrier frequency of a modulated signal given a small number of time-domain samples of that signal; the time domain becomes our observation domain, and the frequency domain becomes the domain in which we wish to determine the unknown shift of the known template signal. For the sake of space, we do not restate all of our results in this context, although the conclusion is clear: a compressive matched filter can successfully estimate an unknown modulation frequency $\omega_{0}$ from $m$ random time-domain samples (with high probability, and within a prescribed tolerance) as long as $m$ scales inversely with the signal-to-noise ratio and logarithmically in the observation duration $|T|$ and the possible range of carrier frequencies $|\Omega|$. (A Nyquist-based approach, in contrast, would require a sampling rate linearly proportional to $|\Omega|$ but could tolerate somewhat lower signal-to-noise ratios.) The bounds will also depend --- in this case --- on the degree to which the template signal is concentrated in the time domain. Signals whose envelope is relatively flat across $T$ will require the fewest measurements, while signals with highly peaked envelopes will require the most.

While we do not discuss this problem further in full generality, we do briefly examine a special case, namely the problem of estimating the frequency of a pure sinusoidal tone from noisy time-domain samples. Such a problem is an ideal candidate for the compressive matched filter because a pure sinusoidal tone has a perfectly flat envelope in the time domain. We discuss this tone estimation problem in Section~\ref{sec:Pure-Tone-Estimation}, and carefully quantify (in Corollaries~\ref{corl:Pmaintone} and~\ref{corl:tonenoise}) the number of random time-domain samples required to successfully estimate the tone's frequency. We also address an important practical question: at how many points is it necessary to sample (or query) the random process when searching for its peak? Using an adaptation of Corollary~\ref{corl:main}, we show for the noiseless case that the empirical peak from a finite set of samples of the random process (with sufficiently dense sampling) must occur within a certain distance of the true peak of the continuous random process. One can therefore employ a grid search strategy for implementing the compressive matched filter, and from the empirical maximum on this grid, one can actually employ a local concave ascent to find the exact value for $\omega_{0}$. We close Section~\ref{sec:Pure-Tone-Estimation} with a stylized application illustrating the potential of extending this work to the problem of determining the arrival time of a linear chirp.

\subsection{Related Work}

To the best of our knowledge, our framework for studying the compressive
matched filter is novel. Prior statistical analysis for compressive
inference problems has focused specifically on problems of signal
detection or classification from a finite model set~\cite{haupt2006compressive,davenport2006sparse,HauptDetect}
or employed a geometric point-of-view based on a stable embedding
of signal family from an original finite-dimensional signal space
into a lower-dimensional measurement space~\cite{davenport2007smashed,davenport2009signal}.
Our work takes a substantially different approach, considering the
inference of a continuous-valued shift parameter from a continuous-time
received signal, and more thoroughly characterizing the statistics
of the problem using the language and tools of empirical processes.

As mentioned above, similar probabilistic tools have been employed
in CS, but for the analysis of the sparse signal recovery problem~\cite{Rudelson-Dudley,ra09,rombergGaussian2009,rombergNeelsh2010,rauhut10re}.
While in principle one could view the matched filter problem as that
of recovering a $1$-sparse signal from a dictionary $\{s_{0}(t-\tau):\tau\in T\}$
of possible candidates, such a dictionary would have infinite size
and extremely high coherence, preventing the application of most standard
recovery analysis techniques. One recent work~\cite{eldarTimeDelay2010}
has formalized the matched filter problem using signal recovery principles
and a union of subspaces model. However, this work is quite different
from ours in that it does not theoretically study noise sensitivity
and relies on a non-random sampling architecture that is carefully
designed to facilitate the solution of the recovery problem. Interestingly,
outside the field of CS, very similar random processes to those that
we study have also arisen in the analysis of the spectral norm of
random Toeplitz matrices~\cite{meckes-spectral}.

The second part of this paper adapts our analysis of the compressive
matched filter to the problem of estimating the frequency of a pure
sinusoidal tone from a small number of random time-domain samples.
The recovery of signals that are sparse in the frequency domain based
on compressive measurements is a problem that has been well-studied
in the CS literature, although most work in this area has been concerned
with signals that can be written as trigonometric polynomials~\cite{candes2006robust,kunis06ra,gilbert2008tutorial,tropp2010beyond}.
Some techniques for recovering
off-grid frequency-sparse signals have been proposed that involve
windowing~\cite{tropp2010beyond} or other classical techniques from
the field of spectral estimation~\cite{Duarte}, and other work has
considered the more general problem of recovering continuous-time
signals based on a union of subspaces model~\cite{eldar2009compressed},
but the analysis that we present is more sharply focused on the statistics
of the simpler pure tone estimation problem.

Finally, we would like to point out some of the differences between the tone estimation problem considered in this paper and the classical problem of estimating the power spectrum of a random process from samples at random locations (see~\cite{beutler1970alias,masry1978poisson,lii1994spectral,bremaud2002power}).  In Sections~\ref{sec:Matched filter noise free} and~\ref{sec:Pure-Tone-Estimation}, we will show how the output of the compressive matched filter is a random process whose mean is the template autocorrelation function. This random process is completely specified by the samples we have observed, and rather than merely estimating its second-order statistics, we will be interested in establishing a uniform bound on its deviation from the template; this will allow us to conclude that it peaks at or near the correct location. It is also worth mentioning that our work differs from Rife and Boorstyn's classical analysis of the single-tone parameter estimation problem~\cite{rife1974single}. Specifically, our work permits sampling below the Nyquist rate, and with high probability we provide an absolute bound on the accuracy of the frequency estimate, rather than involving the Cram\'{e}r-Rao bound.


\section{\label{sec:Matched filter noise free}Analytical Framework and Main Results}


\subsection{Problem Statement}
\label{sec:problem}

\subsubsection{Signal Model}

Suppose we have received a signal $A\cdot s_{0}(t-\tau_{0})$, where
$s_{0}(t)$ is a known signal template, and $\tau_{0}$ and $A$ are
the unknown delay and amplitude, respectively. We assume that the
unknown delay $\tau_{0}$ (also called the {\em time-of-arrival})
is restricted to some interval $T=[\tau_{\mathrm{min}},\tau_{\mathrm{max}}]\subset\R$.
We make no particular assumptions about $s_{0}$, although our bounds will depend on the properties of $s_{0}$ over the range of frequencies where it is observed.

We will consider two closely related cases in this paper: in the {\em
real case}, we restrict both $s_{0}$ and $A$ to be real-valued,
whereas in the {\em complex case}, we allow both $s_{0}$ and $A$
to be complex-valued. Much of our analysis will be identical for the
real and complex cases, but we will specialize our discussions to
distinguish between the two cases when necessary.

\subsubsection{Observations}

We would like to estimate $\tau_{0}$ and $A$ based on random samples of the Fourier transform of the received signal. In particular, we
suppose that we acquire $m$ samples of the Fourier transform of $A\cdot s_{0}(t-\tau_{0})$ at frequencies $\omega_{1},\omega_{2},...,\omega_{m}$, which are drawn independently at random from a uniform distribution on some interval $\Omega=[-\omega_{\mathrm{max}},\omega_{\mathrm{max}}]$ in the frequency domain. Typically, one would choose $\Omega$ roughly equal to the essential bandwidth of $s_0$, although this is not strictly necessary; we more carefully discuss the implications of choosing $\Omega$ in Section~\ref{sec:Robustness} below.

The vector of observations $y\in\mathbb{C}^{m}$ is formed as \begin{equation}
y[k]=A\int_{-\infty}^{\infty}s_{0}(t-\tau_{0})e^{-\iunit\omega_{k}t}\, dt=A\cdot e^{-\iunit\omega_{k}\tau_{0}}\widehat{s}_{0}(\omega_{k}),\quad k=1,2,\ldots,m,\label{eq:yclean}\end{equation}
where $\widehat{s}_{0}(\omega)$ denotes the Fourier transform of
$s_{0}(t)$. For the moment, we assume that these observations are
noiseless. In Section~\ref{sec:Robustness} we extend our formulation
to account for measurement noise.

We define $s(t)$ to be a low-pass filtered version
of $s_{0}(t)$ having frequency content bandlimited to the interval
$\Omega$. More formally, \[
s(t):=\frac{1}{2\pi}\int_{\Omega}\widehat{s}_{0}(\omega)e^{\iunit\omega t}\, d\omega.\]
It follows that $\widehat{s}(\omega)=\widehat{s}_{0}(\omega)$ for
all $\omega\in\Omega$ and that $\widehat{s}(\omega)=0$ for all $\omega\notin\Omega$.
Thus, because our Fourier-domain observations are limited to the interval
$\Omega$, we may rewrite the expression (\ref{eq:yclean}) for our
observations as \[
y[k]=A\cdot e^{-\iunit\omega_{k}\tau_{0}}\widehat{s}(\omega_{k}),\quad k=1,2,\ldots,m.\]
Consequently, all of our subsequent analysis will depend only on
properties of the bandlimited signal $s(t)$.

\subsubsection{Least-Squares Estimation}

Given the observation vector $y$, a natural approach to estimating
$\tau_{0}$ and $A$ is to find the delay and amplitude which best
explain the measurements in a least-squares sense. (Such a least-squares estimate coincides with the maximum likelihood estimate in the case of Gaussian measurement noise, as we consider in Section~\ref{sec:Robustness}.) More formally,
we define \begin{equation}
(\widehat{\tau}_{0},\widehat{A}):=\arg\min_{{\tau},{A}}\sum_{k=1}^{m}\left|y[k]-{A}\cdot e^{-\iunit\omega_{k}{\tau}}\widehat{s}(\omega_{k})\right|^{2}=\arg\min_{{\tau},{A}}\|y-{A}\psi_{{\tau}}\|_{2}^{2},\label{eq:lsmain}\end{equation}
where for any $\tau\in T$, the test vector $\psi_{\tau}\in\mathbb{C}^{m}$
is given by: \[
\psi_{\tau}[k]=e^{-\iunit\omega_{k}\tau}\widehat{s}(\omega_{k}),\quad k=1,2,\ldots,m.\]
For a given estimate $\widehat{\tau}_{0}$ of the delay, one can derive a closed form expression for the amplitude $A$ that minimizes (\ref{eq:lsmain}):
\begin{equation}
\widehat{A}=\left\{ \begin{array}{ll}
\frac{\mathrm{Re}\langle y,\psi_{\widehat{\tau}_{0}}\rangle}{\|\psi_{\widehat{\tau}_{0}}\|_{2}^{2}}, & \mathrm{real~case},\\[2mm]
\frac{\langle y,\psi_{\widehat{\tau}_{0}}\rangle}{\|\psi_{\widehat{\tau}_{0}}\|_{2}^{2}}, & \mathrm{complex~case}.\end{array}\right. \label{eq:newA}
\end{equation}
Plugging (\ref{eq:newA}) in to (\ref{eq:lsmain}), we see that the optimal time-of-arrival estimate is given by
\begin{equation*}
\widehat{\tau}_0 = \left\{ \begin{array}{ll}
\arg\min_{\tau\in T}\frac{-\left|\mathrm{Re}(\langle y,\psi_{\tau}\rangle)\right|^2}{\| \psi_\tau \|^2}, & \mathrm{real~case},\\[1mm]
\arg\min_{\tau\in T}\frac{-\left|\langle y,\psi_{\tau}\rangle\right|^2}{\| \psi_\tau \|^2}, & \mathrm{complex~case}.\end{array}\right.
\end{equation*}
Finally, noting that $\| \psi_\tau\|$ is constant over all $\tau \in T$, we obtain a simplified expression for the least-squares estimate of $\tau_{0}$:
\begin{equation}
\widehat{\tau}_{0}=\left\{ \begin{array}{ll}
\arg\max_{\tau\in T}\left|\mathrm{Re}\langle y,\psi_{\tau}\rangle\right|, & \mathrm{real~case},\\[1mm]
\arg\max_{\tau\in T}\left|\langle y,\psi_{\tau}\rangle\right|, & \mathrm{complex~case}.\end{array}\right.\label{eq:bothls}
\end{equation}

Equation (\ref{eq:bothls}) suggests a correlation-based strategy
for estimating $\tau_{0}$; this strategy is a natural generalization
of the traditional time-domain {}``matched filter'' to our measurement
setting. For this reason, we refer to such an estimator as a {\em
compressive matched filter}, and our focus in this paper will be
on the accuracy with which $\tau_{0}$ can be estimated using such
an estimator. Because $\widehat{A}$ is subsequently
defined in terms of $\widehat{\tau}_{0}$, one could easily
extend our analysis to bound the accuracy of estimating $A$.


\subsection{Noiseless Analysis}

\label{sub:mainclean}

In order to study the performance of a correlation-based estimator
for $\tau_{0}$, let us define the complex-valued random process $X(\tau)$
on $T$ to be the correlation of the observations $y$ with each of
the test vectors $\psi_{\tau}$, \begin{equation}
X(\tau):=\langle y,\psi_{\tau}\rangle=A\sum_{k=1}^{m}\left|\widehat{s}\left(\omega_{k}\right)\right|^{2}e^{\iunit\omega_{k}\left(\tau-\tau_{0}\right)}.\label{eq:X}\end{equation}
 This random process has mean function \begin{equation}
\small\E X(\tau)=A\sum_{k=1}^{m}\E\left|\widehat{s}\left(\omega_{k}\right)\right|^{2}e^{\iunit\omega_{k}(\tau-\tau_{0})}=\frac{Am}{|\Omega|}\int_{\Omega}\left|\widehat{s}\left(\omega\right)\right|^{2}e^{\iunit\omega(\tau-\tau_{0})}~d\omega=\frac{2\pi Am}{|\Omega|}R_{ss}(\tau-\tau_{0}),\label{eq:MF mean function}\end{equation}
where $R_{ss}(\cdot)=\left(s(t)\star s^{*}(-t)\right)(\cdot)$ denotes
the autocorrelation function of $s(t)$.

In the complex case, the compressive matched filter estimate (\ref{eq:bothls})
for $\tau_{0}$ can be interpreted as a search for the maximizer of
$|X(\tau)|$. Because $|R_{ss}(\cdot)|$ is maximized at the origin, one would
expect informally that, on average, finding the maximum magnitude of the process
$X(\tau)$ should correctly estimate $\tau_{0}$. In the real case, the
compressive matched filter estimate for $\tau_{0}$ can be interpreted
as a search for the maximizer of $|\mathrm{Re}(X(\tau))|$. However,
in this case we note that since $R_{ss}(\cdot)$ is real, we will
have $\E\;\mathrm{Re}(X(\tau))=\mathrm{Re}(\E X(\tau))=\E X(\tau)$,
which again has magnitude maximized at $\tau_{0}$, and so informally, finding the maximum magnitude of the process $\mathrm{Re}(X(\tau))$ should correctly estimate $\tau_{0}$.

An equivalent, and potentially more revealing, way to frame the delay
estimation problem is to observe that we could rescale $X(\tau)$
to obtain an estimate of the ideal autocorrelation function $R_{ss}(\cdot)$
(up to the unknown amplitude $A$ and translation $\tau_{0}$). Define
\begin{equation}
\widetilde{R}_{ss}(\tau):=\left\{ \begin{array}{ll}
\frac{|\Omega|}{2\pi m}\cdot\mathrm{Re}(X(\tau)), & \mathrm{real~case},\\[1mm]
\frac{|\Omega|}{2\pi m}\cdot X(\tau), & \mathrm{complex~case}.\end{array}\right.\label{eq:acfestboth}\end{equation}
One way to interpret this estimate is that we have approximated the scaled,
shifted autocorrelation function, $AR_{ss}(\tau-\tau_{0})$, as a
discrete sum with samples taken at random locations in $\Omega$;
equation \eqref{eq:MF mean function} tell us that this estimate is
unbiased, since $\mathbb{E}[\widetilde{R}_{ss}(\tau)]=AR_{ss}(\tau-\tau_{0})$.

It is clear that solving the least-squares problem (\ref{eq:bothls})
is equivalent to finding the maximum of $|\widetilde{R}_{ss}(\tau)|$.
Our main concern will be quantifying how close the random process
$\widetilde{R}_{ss}(\tau)$ is to its mean. 
It is worth noting that, if the measurements are perfectly clean and we are able to perform
all computations to infinite precision, then $|\widetilde{R}_{ss}(\tau)|$
is actually guaranteed to peak at $\tau_{0}$, where it takes its maximum value
of $|\Omega|(|A|2\pi m)^{-1}\|y\|_{2}^{2}$. But what Theorems~\ref{th:Emain}
and \ref{th:Pmain} will tell us is that if $m$ is large enough, then
there will be a tangible gap between this peak at $\tau_{0}$ and
the values of $|\widetilde{R}_{ss}(\tau)|$ for all $\tau$ bounded
some distance away from $\tau_{0}$. As we will then see in Section~\ref{sec:Robustness}, this gap will
make the maximizer of $|\widetilde{R}_{ss}|$ a robust estimate of $\tau_{0}$
in the presence of noise.

To simplify some of the notation, we will use $\eta$ to denote the
peak magnitude of the mean function $AR_{ss}(\tau-\tau_{0})$, \[
\eta= \left|A R_{ss}(0)\right| =\left|A\right|\|s\|_{2}^{2}=\left|A\right|\frac{\|\widehat{s}\|_{2}^{2}}{2\pi}.\]
 Our results will depend on how concentrated the Fourier transform
$\widehat{s}(\omega)$ is over the sampling domain $\Omega$. Intuitively,
if $\widehat{s}$ is spread out more or less evenly over $\Omega$,
then each sample will give us some information about the return signal.
If $\widehat{s}$ is concentrated on a small set within $\Omega$,
then only a small number of the randomly chosen samples will tell
us anything at all. We will quantify this concentration in two different
ways. We introduce \[
\mu_{1}=\frac{\sqrt{|\Omega|}\cdot\|\widehat{s}\|_{4}^{2}}{\|\widehat{s}\|_{2}^{2}},\quad\text{and}\quad\mu_{2}=\frac{|\Omega|\cdot\|\widehat{s}\|_{\infty}^{2}}{\|\widehat{s}\|_{2}^{2}}.\]
If the energy of $\widehat{s}$ is equally spread over the sampling
domain $\Omega$, that is if $|\widehat{s}(\omega)|=|\Omega|^{-1/2}\|\widehat{s}\|_{2}$
for all $\omega\in\Omega$, then it is easy to see that $\mu_{1}=\mu_{2}=1$.
If most of the energy in $\widehat{s}(\omega)$ is concentrated on
a small subset of $\Omega$, then $\mu_{1}$ and $\mu_{2}$ will be
large (and in fact, they can be made arbitrarily large).

We start by getting a rough idea of how close $\widetilde{R}_{ss}(\tau)$
is to its mean by looking at the variance at a shift $\tau$. Since
$|\mathrm{Re}(X(\tau))|\le|X(\tau)|$, we can bound the second moment
of $\widetilde{R}_{ss}(\tau)$ in both the real and complex cases:
\begin{align*}
\E\left[|\widetilde{R}_{ss}(\tau)|^{2}\right] & \le\frac{|\Omega|^{2}}{4\pi^{2}m^{2}}\E\left|X(\tau)\right|^{2}\\
& =\frac{|A|^{2}|\Omega|^{2}}{4\pi^{2}m^{2}}\sum_{k_{1}=1}^{m}\sum_{k_{2}=1}^{m}\E\left[|\s(\omega_{k_{1}})|^{2}|\s(\omega_{k_{2}})|^{2}e^{-\iunit\omega_{k_{1}}(\tau-\tau_{0})}e^{\iunit\omega_{k_{2}}(\tau-\tau_{0})}\right]\\
&
= \frac{|A|^{2}|\Omega|^{2}}{4\pi^{2}m^{2}}\sum_{k_{1}=1}^{m} \E\left[|\s(\omega_{k_{1}})|^{4}\right] \\ & ~~~  + \frac{|A|^{2}|\Omega|^{2}}{4\pi^{2}m^{2}}\sum_{k_{1}=1}^{m} \sum_{k_{2} \neq k_{1}}\E_{\omega_{k_1}}\left[|\s(\omega_{k_{1}})|^{2}e^{-\iunit\omega_{k_{1}}(\tau-\tau_{0})}\right] \E_{\omega_{k_2}}\left[|\s(\omega_{k_{2}})|^{2}e^{\iunit\omega_{k_{2}}(\tau-\tau_{0})}\right] \\
&
= \frac{|A|^{2}|\Omega|^{2}}{4\pi^{2}m^{2}}\sum_{k_{1}=1}^{m} \frac{1}{|\Omega|} \int_\Omega |\s(\omega)|^{4} \; d\omega  \\ & ~~~ +  \frac{|A|^{2}|\Omega|^{2}}{4\pi^{2}m^{2}}\sum_{k_{1}=1}^{m}  \sum_{k_{2} \neq k_{1}} \frac{1}{|\Omega|^2} \left(\int_\Omega |\s(\omega)|^{2} e^{\iunit\omega (\tau-\tau_{0})} \; d\omega \right)^* \left(\int_\Omega |\s(\omega)|^{2} e^{\iunit\omega (\tau-\tau_{0})} \; d\omega \right) \\
& =\frac{|A|^{2}|\Omega|^{2}}{4\pi^{2}m^{2}}\left(\frac{m}{|\Omega|}\|\widehat{s}\|_{4}^{4}+\frac{m(m-1)4\pi^{2}}{|\Omega|^{2}}\left|R_{ss}(\tau-\tau_{0})\right|^{2}\right)\\
&
= \frac{|A|^{2}|\Omega|\|\widehat{s}\|_{4}^{4}}{4\pi^{2}m}+\frac{m-1}{m}|A|^2|R_{ss}(\tau-\tau_0)|^2,
\end{align*}
whereas $|\E[\widetilde{R}_{ss}(\tau)]|^{2} = |A|^2|R_{ss}(\tau-\tau_0)|^2$. Therefore,
\begin{align*}
\operatorname{Var}\left[\widetilde{R}_{ss}(\tau)\right]=\E\left[|\widetilde{R}_{ss}(\tau)|^{2}\right]-\left|\E[\widetilde{R}_{ss}(\tau)]\right|^{2}\leq\frac{\eta^{2}}{m}\left(\mu_{1}^{2}-\frac{4\pi^{2}|R_{ss}(\tau-\tau_{0})|^{2}}{\|\widehat{s}\|_{2}^{4}}\right)\leq\frac{\eta^{2}\mu_{1}^{2}}{m}.
\end{align*}
Using Jensen's inequality, we then obtain a bound for the expected deviation
of $\widetilde{R}_{ss}(\tau)$ from its mean at a fixed shift $\tau$:
\begin{equation}
\E\left|\widetilde{R}_{ss}(\tau)-AR_{ss}(\tau-\tau_{0})\right|~\leq~\frac{\eta\mu_{1}}{\sqrt{m}}.\label{eq:Epoint}\end{equation}
As expected, this deviation gets smaller as $m$ increases, and scales with $\mu_{1}$.

Our first theorem gives us a {\em uniform} bound for the expected
maximum deviation of $\widetilde{R}_{ss}(\tau)$ from its mean over
all $\tau\in T$. The following result is proved in Section~\ref{sec:mainproofs}.
\begin{thm}
	\label{th:Emain}
	Suppose that $\TO\ge 3$.\footnote{If $\TO < 3$, this theorem and all of our bounds still hold but with weaker constants.} Then the autocorrelation function estimate $\widetilde{R}_{ss}(\tau)$, as defined in (\ref{eq:acfestboth}), obeys
	\begin{align}
		\mathbb{E}\sup_{\tau\in T}
		\left|\widetilde{R}_{ss}(\tau)-AR_{ss}(\tau-\tau_{0})\right|&
		\leq \frac{\eta\mu_{1}}{\sqrt{m}}\cdot\left(4.25\sqrt{\log(2\TO)}+2.28\right)\nonumber\\
		&\leq 5.96\cdot\frac{\eta\mu_{1}}{\sqrt{m}}\cdot\sqrt{\log(2\TO)}.
		\label{eq:Emain}
	\end{align}
\end{thm}

The essential difference between \eqref{eq:Epoint} and \eqref{eq:Emain} is the
factor of $\sqrt{\log(2|\Omega||T|)}$ --- this is the price we are
paying for a bound which holds uniformly over all $\tau\in T$. The
bound slowly loosens as the time-bandwidth product $|T||\Omega|$
gets larger; this effect is weak but necessary, as $|T||\Omega|$
affects the complexity of the random process. (A similar penalty arises
in standard CS bounds~\cite{candes2006robust,donoho2006compressed},
where the number of measurements required for successful, robust signal recovery
grows logarithmically with the ambient dimension of the signal space --- this logarithmic dependence is known to be sharp.) 

We note that Theorem~\ref{th:Emain} could be proved using the Dudley inequality~\cite{dudley1967sizes}, a classical tool which relates the supremum of a random process to the geometry of its index set; the main challenges arise in computing covering numbers for the index set $T$ under certain metrics defined in terms of the random process $\widetilde{R}_{ss}(\tau)$ and in adapting the Dudley argument to complex numbers. To provide better insight and to obtain sharper constants, however, our proof in Section~\ref{sec:mainproofs} more directly customizes the derivation of the Dudley inequality for our particular scenario. We also note that a simple application of the Sudakov minoration principle~\cite{LedouxBanachSpaces} (after computing the necessary metrics) reveals that the bound in Theorem~\ref{th:Emain} is indeed sharp (up to a constant). Intuitively, $|\Omega||T|$ is the number of points on a grid of resolution $1/|\Omega|$ on $T$ necessary to control the deviations of the random process.

Theorem~\ref{th:Emain} demonstrates that $\widetilde{R}_{ss}(\tau)$
is close to its mean in expectation; our second theorem demonstrates
that it is also close with high probability. The following is proved in Section~\ref{sec:mainproofs}.
\begin{thm} \label{th:Pmain}
	Fix $\delta>0$ and let
	\begin{equation} U = C_1 \cdot
    \max\left(\frac{\eta\mu_{1}}{\sqrt{m}},~~\frac{\eta\mu_{2}}{m}\cdot\sqrt{\log(4/\delta)}\right)\cdot\sqrt{\log(12\TO/\delta)},\label{eq:U}
	\end{equation}
    where $C_1$ is a known universal constant.
	If $\left|\Omega\right|\left|T\right|\ge3$,
	then the autocorrelation function estimate $\widetilde{R}_{ss}(\tau)$,
	as defined in (\ref{eq:acfestboth}), obeys
	\begin{equation}
		  \P{ \sup_{\tau\in T}\left|\widetilde{R}_{ss}(\tau)-AR_{ss}(\tau-\tau_{0})\right|>U} ~\leq~\delta.\label{eq:pconc}
	\end{equation}
\end{thm}


\subsection{Example: A Gaussian Pulse}

\label{sec:gaussianpulse}

A concrete example will help illustrate what Theorems~\ref{th:Emain}
and~\ref{th:Pmain} are telling us about the effectiveness of the
compressive matched filter. Suppose that $s_{0}(t)$ is a real-valued
Gaussian pulse with unit energy, \begin{equation}
s_{0}(t)=\pi^{-1/4}a^{-1/2}e^{-t^{2}/2a^{2}},\label{eq:gaussian}\end{equation}
 We will assume that this pulse is received with a time-of-arrival
$\tau_{0}$ in the interval $T=[0,1]$, and that it is scaled by an
unknown real-valued amplitude $A$. We will also assume that the width
$a$ of the pulse is much less than $1$, and so to estimate $\tau_{0}$
reliably from samples in the time domain, we would need on the order
of $1/a$ samples on $T$. Figure~\ref{fig:s0gauss}(a) shows
an example received signal $A \cdot s_{0}(t-\tau_{0})$ for $A=1$, $\tau_{0}=0.4$
and $a=1/200$.

The Fourier transform of $s_{0}$ is \[
\widehat{s}_{0}(\omega)=\sqrt{2a}\pi^{1/4}e^{-a^{2}\omega^{2}/2}.\]
 We will take as our sampling domain $\Omega=[-3/a,3/a]$; $\widehat{s}$
is simply $\widehat{s}_{0}$ bandlimited to this interval. The bandlimited
signal $\widehat{s}$ is nearly identical to $\widehat{s}_{0}$; we
can calculate \[
\|\widehat{s}_{0}\|_{2}^{2}=2\pi,\quad\|\widehat{s}_{0}\|_{4}^{2}=(2\pi)^{3/4}\sqrt{a},\]
 and standard bounds for integrating the tails of $e^{-a^{2}\omega^{2}}$
show us that $\|\widehat{s}\|_{2}^{2}$ is within $1.4\cdot10^{-4}$
of $\|\widehat{s}_{0}\|_{2}^{2}$ and $\|\widehat{s}\|_{4}^{2}$ is
within $10^{-8}$ of $\|\widehat{s}_{0}\|_{4}^{2}$. We can safely
say that $\mu_{1}\leq1.6$ and $\mu_{2}\leq3.4$ (these values are
the same for all choices of $a$). The Fourier transform for $a=1/200$
over the range $\Omega=[-600,600]$ is shown in Figure~\ref{fig:s0gauss}(b).

For a received signal with parameters $A=1$, $\tau_{0}=0.4$, and $a=1/200$, Figure~\ref{fig:Rsshat} shows the estimate $\widetilde{R}_{ss}(\tau)$ of the scaled, shifted autocorrelation function based on $m=10$, $20$, and $50$ random frequency domain samples, along with the true scaled, shifted autocorrelation function $A\cdot R_{ss}(\tau-\tau_{0})\approx e^{-\left(\tau-\tau_{0}\right)^{2}/(4a^{2})}$. In all cases, we see that $|\widetilde{R}_{ss}(\tau)|$ reaches its peak exactly at $\tau_{0}$; as noted in Section~\ref{sub:mainclean}, this is to be expected in the case of noiseless measurements. However, we also see a gap between the peak at $\tau_{0}$ and the remainder of the estimate that becomes larger as $m$ increases.

To see how this behavior is supported by our theory, note that for a Gaussian pulse with $A=1$, we know that the mean function $R_{ss}(\tau-\tau_{0})=1$ for $\tau=\tau_{0}$ and $R_{ss}(\tau-\tau_{0})\leq0.1054$ for $|\tau-\tau_{0}|>3a$. (For simplicity, these calculations assume $s = s_{0}$ exactly.) If $U$ is the value from
\eqref{eq:U} in Theorem~\ref{th:Pmain}, then we are guaranteed
that the difference between the peak value $\widetilde{R}_{ss}(\tau_{0})$
and any $\widetilde{R}_{ss}(\tau)$ for $|\tau-\tau_{0}|\ge3a$ is
at least $\epsilon$ when $1-U\geq0.1054+U+\epsilon$, i.e., when $U\leq\frac{0.8946-\epsilon}{2}$.
Note that $U$ can be made small enough for large enough $m$, namely $m\gtrsim\log(1/a)$. Thus using the compressive matched filter, we can reliably infer the time-of-arrival
from $\sim\log(1/a)$ randomly chosen samples in the frequency domain
as opposed to $\sim1/a$ equally spaced samples in the time domain.

\begin{figure}[t]
 \centering \begin{tabular}{cc}
\includegraphics[height=1.4in]{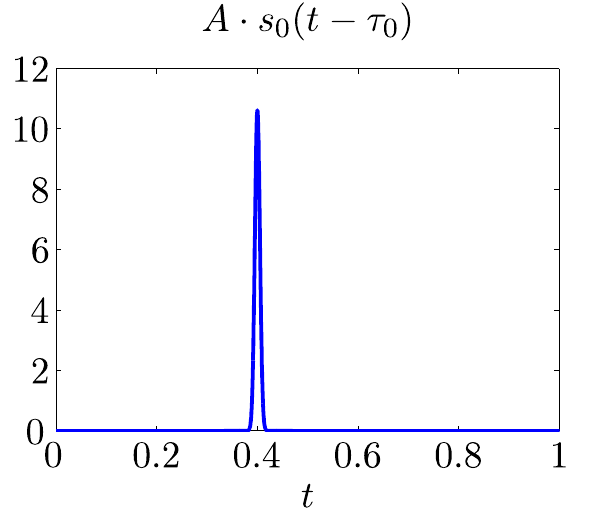}   & \includegraphics[height=1.4in]{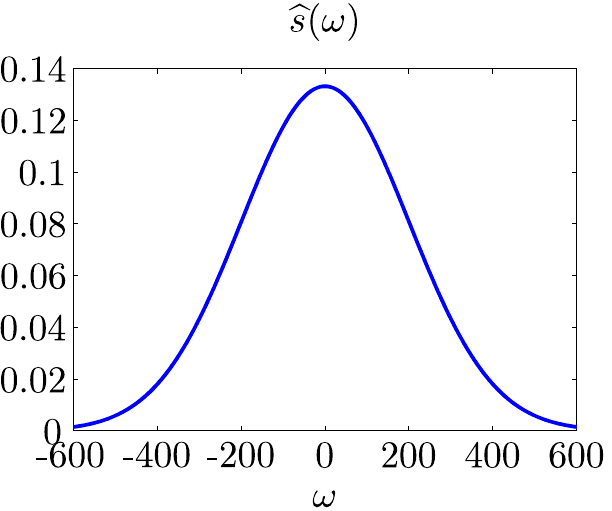} \tabularnewline
\small (a)   & \small (b) \tabularnewline
\end{tabular}\vspace{-2mm}\caption{\sl\small{(a) The return signal for Gaussian pulse from \eqref{eq:gaussian}
with $a=1/200$ and return parameters $A=1$ and $\tau_{0}=0.4$.
(b) The Fourier transform $\widehat{s}(\omega)$ on $\Omega=[-600,600]$.
Since $\widehat{s}$ is relatively diffuse over $\Omega$, both measures
of frequency concentration are not too large: $\mu_{1}\leq1.6,~\mu_{2}\leq3.4$. }}
\label{fig:s0gauss}
\end{figure}

\begin{figure}[t]
 \centering \begin{tabular}{ccc}
\includegraphics[height=1.4in]{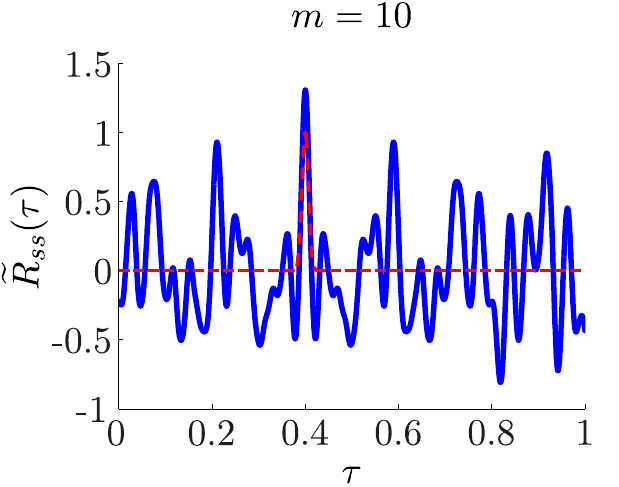}   & \includegraphics[height=1.4in]{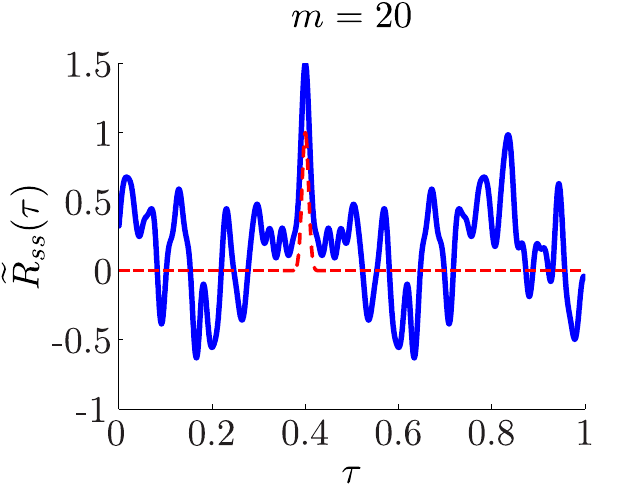}   & \includegraphics[height=1.4in]{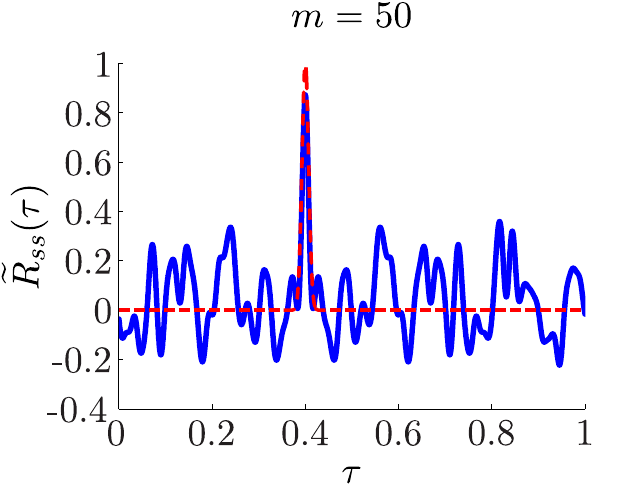} \tabularnewline
\small (a)   & \small (b)   & \small (c) \tabularnewline
\end{tabular}\vspace{-2mm}\caption{\sl\small{Estimated scaled, shifted autocorrelation function $\widetilde{R}_{ss}(\tau)$
(solid blue line) and true scaled, shifted autocorrelation function
$A\cdot R_{ss}(\tau-\tau_{0})$ (dashed red line) for (a) $m=10$,
(b) $m=20$, (c) and $m=50$ measurements. $\widetilde{R}_{ss}(\tau)$
is a random process whose mean is $A\cdot R_{ss}(\tau-\tau_{0})$;
as the number of measurements increases, this process deviates less
from its mean (see Theorems~\ref{th:Emain} and~\ref{th:Pmain}).}}
\label{fig:Rsshat}
\end{figure}


\subsection{General Noiseless Performance Characterization}
\label{sec:gnpc}

Our statements about about quantifying the number of samples needed
to ensure a clear separation between the peak of $|\widetilde{R}_{ss}(\tau)|$
and the function away from the peak are easily generalized. The statements
in this section can be interpreted as a condition on the number of
samples needed to ensure the successful operation of the compressive
matched filter. The result below is interesting when the underlying
autocorrelation function $R_{ss}(\tau)$ has one main peak (a {}``main
lobe'') centered at $\tau=0$, and is relatively small away from
the origin. This situation is typical, but similar statements could
be formulated depending on the assumptions one wishes to impose on
$s(t)$ (and its autocorrelation function).

\begin{corl} \label{corl:main} Suppose there exist constants $\alpha_{1}\in[0,1)$ and $\alpha_{2}>0$ such that $|R_{ss}(\tau)|\le\alpha_{1}R_{ss}(0)$
for all $|\tau|>\alpha_{2}$, and choose $\epsilon \in [0,1-\alpha_1]$. Suppose also that $\left|\Omega\right|\left|T\right|\ge3$
and that
%
\begin{equation}
m > C_2 \cdot \max\left(\frac{\log(12\TO/\delta)}{(1-\alpha_{1}-\epsilon)^2}\cdot \mu_{1}^{2} ,\; \frac{\sqrt{\log(4/\delta)\log(12\TO/\delta)}}{1-\alpha_1-\epsilon}\cdot\mu_{2} \right),\label{eq:mbound2new}
\end{equation}
where $C_2$ is a known universal constant. Then with probability at least $1-\delta$, $|\widetilde{R}_{ss}(\tau_{0})| > |\widetilde{R}_{ss}(\tau)| + \epsilon \eta$ for all $\tau$ such that $|\tau-\tau_{0}| > \alpha_2$.
\end{corl}
\begin{proof} Supposing we have the concentration suggested
by (\ref{eq:pconc}), we will have $|\widetilde{R}_{ss}(\tau_{0})|\ge |A|R_{ss}(0)-U$
and $|\widetilde{R}_{ss}(\tau)|\le\alpha_{1}|A|R_{ss}(0)+U$ for all
$\tau$ such that $|\tau-\tau_{0}|>\alpha_{2}$. If \eqref{eq:mbound2new}
is satisfied with $C_2 = \max(4 C_1^2,2C_1)$, then
$U < \frac{1}{2} |A|R_{ss}(0) (1-\alpha_1-\epsilon)$ and
it follows that $|A|R_{ss}(0)-U > \alpha_{1}|A|R_{ss}(0)+U+\epsilon |A|R_{ss}(0)$.
A slightly stronger version of this corollary also holds if one omits $C_2$ and chooses constants of $4 C_1^2$ and $2C_1$ for the first and second terms in \eqref{eq:mbound2new}, respectively.
\end{proof}

For the case of noiseless measurements, Corollary~\ref{corl:main} ensures that no values of $\tau$ far from $\tau_{0}$ can give $|\widetilde{R}_{ss}(\tau)|$ close to $|\widetilde{R}_{ss}(\tau_{0})|$. This behavior will become particularly relevant in Section~\ref{sec:Robustness}, where we introduce noise into the measurement process.

We can reveal some of the intuition behind the measurement bound (\ref{eq:mbound2new})
by considering three special cases for the signal $s$. First, consider
$s(t)$ for which $|\widehat{s}(\omega)|$ is uniform over $\Omega$.
In this case, we have $|\widehat{s}(\omega)|=|\Omega|^{-1/2}\|\widehat{s}\|_{2}$
for all $\omega\in\Omega$ and so $\mu_{1}=\mu_{2}=1$. This means
that the requisite number of random measurements (\ref{eq:mbound2new})
for successful operation of the compressive matched filter scales
as $m\sim\log(|\Omega||T|)$.

Alternatively, consider the case where $|\widehat{s}(\omega)|$ is
not perfectly uniform over $\Omega$, but rather we assume that for
some $\beta\ge1$ it obeys $|\widehat{s}(\omega)|\le\beta|\Omega|^{-1/2}\|\widehat{s}\|_{2}$
for all $\omega\in\Omega$, and so $\mu_{2}\leq\beta^{2}$. Using
the fact that $\|\widehat{s}\|_{4}^{4}\le\|\widehat{s}\|_{2}^{2}\|\widehat{s}\|_{\infty}^{2}$
gives us the estimate $\mu_{1}\leq \beta$.  Therefore, (\ref{eq:mbound2new})
now demands that $m\sim\beta^{2}\log(|\Omega||T|)$ --- the factor
of $\beta^{2}$ is the price we pay for the non-uniformity of $\widehat{s}$.

As a final example, consider the special case where $s$ is bandlimited
to some interval $\Omega_{B}\subseteq\Omega$ and $|\widehat{s}(\omega)|$
is uniform over $\Omega_{B}$, i.e., \[
|\widehat{s}(\omega)|=\left\{ \begin{array}{ll}
\frac{\|\widehat{s}\|_{2}}{\sqrt{|\Omega_{B}|}}, & \omega\in\Omega_{B}\\
0, & \omega\in\Omega\backslash\Omega_{B},\end{array}\right.\]
 One way we could interpret this situation is that we have chosen
the sampling domain $\Omega$ to be too large. In this case, we have
$\mu_{1}=\sqrt{|\Omega|/|\Omega_{B}|}$ and $\mu_{2}=|\Omega|/|\Omega_{B}|$,
and so (\ref{eq:mbound2new}) now demands that $m\sim(|\Omega|/|\Omega_{B}|)\log(|\Omega||T|)$.
The penalty $|\Omega|/|\Omega_{B}|$ is a natural oversampling factor
since, on average, only $|\Omega_{B}|$ out of every $|\Omega|$ random
Fourier samples will carry any information about the signal.


\subsection{Robustness in the Presence of Measurement Noise}

\label{sec:Robustness}

We now extend our analysis to account for additive complex-valued
noise in our observations. For random frequencies $\{\omega_{k}\}$
taken uniformly from $\Omega$, we assume that the noisy measurement
vector, $y_{n}$, is formed  as \[
y_{n}[k]=Ae^{-\iunit\omega_{k}\tau_{0}}\widehat{s}_{0}(\omega_{k})+n_{k},\quad k=1,2,\ldots,m,\]
 where the additive noise terms $\{n_{k}\}$ are independent zero-mean
complex-valued Gaussian random variables%
\footnote{That is, the real and imaginary parts of each $n_{k}$ are independent,
real-valued zero-mean Gaussian random variables with variance $\frac{\sigma_{n}^{2}}{2}$.%
}{} with variance $\sigma_{n}^{2}$, and the noise vector is $n:=[n_{1},n_{2},...,n_{m}]^{T}$.
Computing the inner product of $y_{n}$ with the test vector $\psi_{\tau}$
for all $\tau\in T$ leads us to the process \begin{equation}
X_{n}(\tau):=\left\langle y_{n},\psi_{\tau}\right\rangle =X(\tau)+N(\tau),\label{eq:xndef}\end{equation}
 where $X(\tau)$ is as defined in Section \ref{sub:mainclean},
and $N(\tau)$ is the noise process that quantifies the effect of
additive noise in our analysis: \[
N(\tau):=\left\langle n,\psi_{\tau}\right\rangle =\sum_{k=1}^{m}n_{k}\widehat{s}_{0}^{*}(\omega_{k})e^{\iunit\omega_{k}\tau}=\sum_{k=1}^{m}n_{k}\widehat{s}^{*}(\omega_{k})e^{\iunit\omega_{k}\tau}.\]
The noise process is zero-mean, i.e., $\E N(\tau)=0$.

Thus, in the case of noisy observations, we can estimate the ideal
autocorrelation function $R_{ss}(\cdot)$ (up to the unknown amplitude
$A$ and translation $\tau_{0}$) simply by rescaling the noisy random
process $X_{n}(\tau)$. Let us define \begin{equation}
\widetilde{N}(\tau):=\left\{ \begin{array}{ll}
\frac{|\Omega|}{2\pi m}\cdot\mathrm{Re}(N(\tau)), & \mathrm{real~case},\\[1mm]
\frac{|\Omega|}{2\pi m}\cdot N(\tau), & \mathrm{complex~case},\end{array}\right.\label{eq:ntildeboth}\end{equation}
 and note that in either case, $\E[\widetilde{N}(\tau)]=0$. Then, if
we set \[
\widetilde{R}_{ss,n}(\tau):=\left\{ \begin{array}{ll}
\frac{|\Omega|}{2\pi m}\cdot\mathrm{Re}(X_{n}(\tau)), & \mathrm{real~case},\\[1mm]
\frac{|\Omega|}{2\pi m}\cdot X_{n}(\tau), & \mathrm{complex~case},\end{array}\right.\]
 it follows in either case that \begin{equation}
\widetilde{R}_{ss,n}(\tau)=\widetilde{R}_{ss}(\tau)+\widetilde{N}(\tau),\label{eq:rrn}\end{equation}
 where $\widetilde{R}_{ss}(\tau)$ is as defined in~(\ref{eq:acfestboth}).
This function provides an unbiased estimate of the shifted, scaled autocorrelation
function of $s(t)$, since $\mathbb{E}[\widetilde{R}_{ss,n}(\tau)]=AR_{ss}(\tau-\tau_{0})$.

We can gain some intuition for how the noise is hindering the estimation
process with a quick estimate on its expected size at a fixed point
$\tau$. In both the real and complex cases, the variance of $\widetilde{N}(\tau)$
is bounded by \begin{align*}
\operatorname{Var}\left[\widetilde{N}(\tau)\right] & \le\frac{|\Omega|^{2}}{4\pi^{2}m^{2}}\E\left|\sum_{k=1}^{m}n_{k}\widehat{s}^{*}(\omega_{k})e^{\iunit\omega_{k}\tau}\right|^{2}\\
 & =\frac{|\Omega|^{2}}{4\pi^{2}m^{2}}\sum_{k_{1}=1}^{m}\sum_{k_{2}=1}^{m}\E[n_{k_{1}}n_{k_{2}}^{*}]\E[\widehat{s}^{*}(\omega_{k_{1}})\widehat{s}(\omega_{k_{2}})e^{\iunit(\omega_{k_{1}}-\omega_{k_{2}})\tau}]\\
 & =\frac{|\Omega|^{2}}{4\pi^{2}m^{2}}\sum_{k=1}^{m}\E|n_{k}|^{2}\E|\widehat{s}(\omega_{k})|^{2}\\
 & =\frac{|\Omega|}{4\pi^{2}m}\sigma_{n}^{2}\|\widehat{s}\|_{2}^{2},\end{align*}
 and so, using Jensen's inequality, we obtain \begin{equation}
\E|\widetilde{N}(\tau)|~\leq~\frac{1}{2\pi}\cdot\sigma_{n}\cdot\frac{\sqrt{|\Omega|}\|\widehat{s}\|_{2}}{\sqrt{m}}.\label{eq:noisepoint}\end{equation}
Recall that the peak of the noiseless estimate $|\widetilde{R}_{ss}(\tau)|$
is on the order of $\left|A\right|R_{ss}(0)=\left|A\right|(2\pi)^{-1}\|\widehat{s}\|_{2}^{2}$.
Thus the noise process will overwhelm the peak of the noiseless estimate
when \begin{equation}
\sigma_{n}~\sim~\left|A\right|\|\widehat{s}\|_{2}\sqrt{\frac{m}{|\Omega|}}.\label{eq:sigmarough}\end{equation}
 Theorems~\ref{thm:noiseexp} and \ref{thm:mainNoisenew} below show
that for $m$ large enough, we will have essentially the same bound
as \eqref{eq:noisepoint} hold {\em uniformly} over the entire
search interval $T$. As a result, the amount of noise (size of $\sigma_{n}^{2}$)
the compressive matched filter can withstand is essentially (to within
constant and $\log$ factors) the same as in \eqref{eq:sigmarough}.

We start with a bound on the expected maximum of the noise process. The following result is proved in Section~\ref{sec:noiseproofs}.
\begin{thm}
	\label{thm:noiseexp}
	Suppose that $\left|\Omega\right|\left|T\right|\ge 3$.  Then the noise process $\widetilde{N}(\tau)$, as defined in (\ref{eq:ntildeboth}), obeys
	\begin{align*}
		\E\sup_{\tau\in T}\left|\widetilde{N}(\tau)\right|
	&\le \sigma_{n}\cdot\frac{\sqrt{|\Omega|}\|\widehat{s}\|_{2}}{\sqrt{m}}\cdot\left({0.199\sqrt{\log(|\Omega||T|)}}+ 0.166\right)\\
    &\le0.36\cdot\sigma_{n}\cdot\frac{\sqrt{|\Omega|}\|\widehat{s}\|_{2}}{\sqrt{m}}\cdot{\sqrt{\log(|\Omega||T|)}}.
	\end{align*}
\end{thm}

The next theorem shows that, given $m$ large enough, the maximum
of the noise process will not be too much larger than its mean with
high probability. The following result is also proved in Section~\ref{sec:noiseproofs}.
\begin{thm} Fix $\delta>0$. Suppose that $\left|\Omega\right|\left|T\right|\ge3$
	and that
	\begin{equation}
		m\geq C_3 \cdot\max\left(\mu_{1}^{2},~\mu_{2}\right)\cdot\log(1/\delta),\label{eq:mbound1noisenew}
	\end{equation}
    where $C_3$ is a known universal constant. 
	Then the noise process
	$\widetilde{N}(\tau)$, as defined in (\ref{eq:ntildeboth}), obeys
	\begin{equation}
		\P{ \sup_{\tau\in T}\left|\widetilde{N}\left(\tau\right)\right|\ge C_4 \cdot \sigma_{n}\cdot\frac{\sqrt{|\Omega|}\|\widehat{s}\|_{2}}{\sqrt{m}}\cdot\max\left(\sqrt{\log(\TO)},\sqrt{\log(2/\delta)}\right)} ~\le~\delta,
	\label{eq:noisetailnew}
	\end{equation}
    where $C_4$ is a known universal constant. 
\label{thm:mainNoisenew}
\end{thm}

These two theorems, taken in conjunction with Theorems~\ref{th:Emain}
and \ref{th:Pmain}, give us a bound on how far the estimate $\widetilde{R}_{ss,n}(\tau)$
created from noisy samples will vary from its mean. With high probability,
we will have
$$
\left|\widetilde{R}_{ss,n}(\tau)-AR_{ss}(\tau-\tau_{0})\right|~\lesssim~
\max\left(\frac{\eta\mu_{1}}{\sqrt{m}},~\frac{\eta\mu_{2}}{m},~\frac{\sigma_{n}\sqrt{|\Omega|}\|\widehat{s}\|_{2}}{\sqrt{m}}\right)\cdot \sqrt{\log(|\Omega||T|)}.
 $$
 Just as in the noiseless case, the bound on this deviation of our
estimate of the autocorrelation function can be translated directly
into a performance guarantee for the compressive matched filter. This
is codified in the following corollary.

\begin{corl} \label{corl:noise} Suppose there exist constants $\alpha_{1}\in[0,1)$
and $\alpha_{2}>0$ such that $|R_{ss}(\tau)|\le\alpha_{1}R_{ss}(0)$
for all $|\tau|>\alpha_{2}$. Suppose also that $\left|\Omega\right|\left|T\right|\ge3$,
that (\ref{eq:mbound1noisenew}) is satisfied, and that 
\begin{align}
m > & \; C_5 \cdot \max\left( \frac{\log(12\TO/\delta)}{(1-\alpha_{1})^2}\cdot \mu_{1}^{2} ,\; \frac{\sqrt{\log(4/\delta)\log(12\TO/\delta)}}{1-\alpha_1}\cdot\mu_{2},\nonumber \right.
\\
& ~~~~~~~~~~~~~ \left. 
\frac{ \max\left(\log(\TO),\log(2/\delta)\right)}{(1-\alpha_{1})^{2}}\cdot\frac{\sigma_{n}^{2}|\Omega|}{|A|^{2}\|\widehat{s}\|_{2}^{2}}  \right),
\label{eq:mbound2newnoise}
\end{align}
where $C_5$ is a known universal constant. Then with probability at least $1-2\delta$,
the maximum value of $|\widetilde{R}_{ss,n}(\tau)|$ must be attained
for some $\widehat{\tau}_{0}$ within the interval $[\tau_{0}-\alpha_{2},\tau_{0}+\alpha_{2}]$.
\end{corl} \begin{proof} Using (\ref{eq:rrn}), we have that \[
\sup_{\tau}|\widetilde{R}_{ss,n}(\tau)-AR_{ss}(\tau-\tau_{0})|\le\sup_{\tau}|\widetilde{R}_{ss}(\tau)-AR_{ss}(\tau-\tau_{0})|+\sup_{\tau}|\widetilde{N}(\tau)|,\]
 where $\widetilde{R}_{ss}(\tau)$ is defined in (\ref{eq:acfestboth}).
With probability at least $1-2\delta$, both
(\ref{eq:pconc}) and (\ref{eq:noisetailnew}) will be satisfied,
and so we will have \begin{equation}
|\widetilde{R}_{ss,n}(\tau_{0})|\ge\left|A\right|\frac{\|\widehat{s}\|_{2}^{2}}{2\pi}-U-\frac{C_4  \sigma_{n}\sqrt{|\Omega|}\|\widehat{s}\|_{2}\max\left(\sqrt{\log(\TO)},\sqrt{\log(2/\delta)}\right)}{\sqrt{m}}
\label{eq:ncorl1}\end{equation}
 and \begin{equation}
|\widetilde{R}_{ss,n}(\tau)|\le\alpha_{1}\left|A\right|\frac{\|\widehat{s}\|_{2}^{2}}{2\pi}+U+\frac{C_4 \sigma_{n}\sqrt{|\Omega|}\|\widehat{s}\|_{2}\max\left(\sqrt{\log(\TO)},\sqrt{\log(2/\delta)}\right)}{\sqrt{m}}
\label{eq:ncorl2}\end{equation}
 for all $\tau$ such that $|\tau-\tau_{0}|>\alpha_{2}$. If \eqref{eq:mbound2newnoise}
is satisfied with $C_5 = \max(16 C_1^2, 4C_1, 64\pi^2 C_4^2)$, then it follows that the right hand side of (\ref{eq:ncorl1})
must exceed the right hand side of (\ref{eq:ncorl2}). As with Corollary~\ref{corl:main}, one can slightly strengthen this result by choosing differing constants for the three terms appearing in \eqref{eq:mbound2newnoise}.
\end{proof}

To within a constant factor, the first two terms in \eqref{eq:mbound2newnoise} are the same as
in Corollary~\ref{corl:main}; we might think of these terms as {}``activation''
conditions for when the compressive matched filter will be well-behaved
in the absence of noise. After these conditions are met, it can withstand
noise levels up to a size \begin{equation}
\sigma_{n}~\sim~\left|A\right|\|\widehat{s}\|_{2}\sqrt{\frac{m}{|\Omega|\log(|\Omega||T|)}}.\label{eq:sigprecise}\end{equation}
 We can interpret \eqref{eq:sigmarough} as the noise level at which
the operation of the compressive matched filter will fall apart completely,
and \eqref{eq:sigprecise} (which is only a log factor smaller) as
the noise level at which we have guaranteed accuracy.

Three examples of the estimated autocorrelation function $\widetilde{R}_{ss,n}(\tau)$
are shown in Figure~\ref{fig:Rssnoise} for the Gaussian pulse example
from Section \ref{sec:gaussianpulse} with a fixed number $m=50$
of samples and various values of $\sigma_{n}^{2}$ (the noiseless
estimates $\widetilde{R}_{ss}(\tau)$ are overlaid). For the same
number of measurements, Figure~\ref{fig:Pecentage-of-correct} shows
the average performance of the compressive matched filter versus the
noise level. For various noise levels $\sigma_{n}$ between $0$ and
$|A|\|\widehat{s}\|_{2}\sqrt{m/|\Omega|}$, we run 1000 experiments
generating random sample frequencies and random noise, and estimate
$\tau_{0}$ by identifying the peak of $|\widetilde{R}_{ss,n}(\tau)|$.
The figure indicates the percentage of trials in which the delay was
estimated to within a distance $2a$ of the correct value $\tau_{0}$.
We see in these experiments that the estimator begins to lose effectiveness
roughly when $\sigma_{n}\sim0.25\cdot |A|\|\widehat{s}\|_{2}\sqrt{m/|\Omega|}$.

\begin{figure}[t]
\centering \begin{tabular}{ccc}
\includegraphics[height=1.4in]{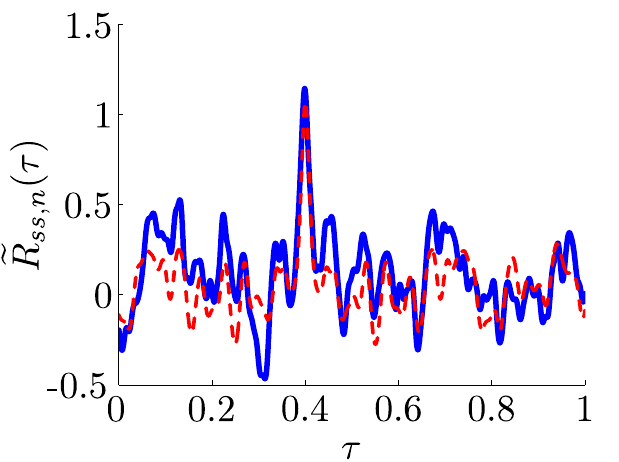}   & \includegraphics[height=1.4in]{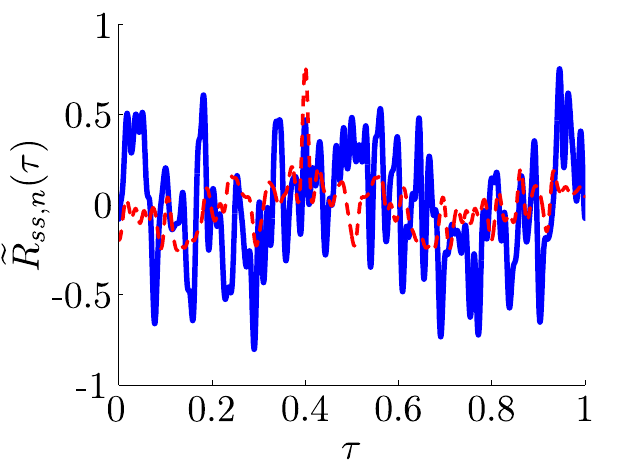}   & \includegraphics[height=1.4in]{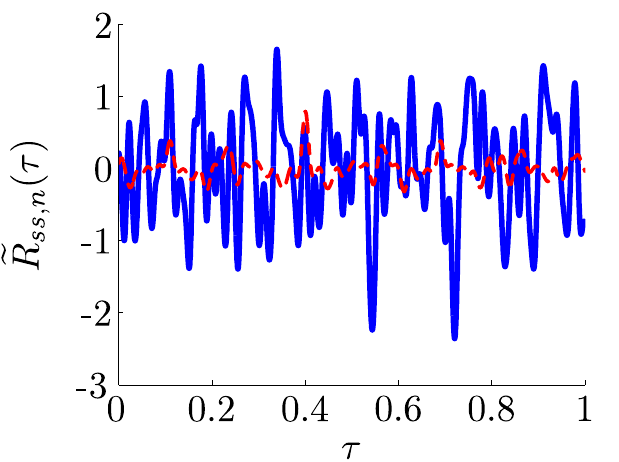} \tabularnewline
\small (a)   & \small (b)   & \small (c) \tabularnewline
\end{tabular}\vspace{-2mm}\caption{\sl\small{Estimated scaled, shifted autocorrelation function $\widetilde{R}_{ss,n}(\tau)$
(solid blue line) obtained from noisy samples, and for the sake of
comparison, the estimate $\widetilde{R}_{ss}(\tau)$ (dashed red line)
that would have been obtained without noise. For all experiments,
the number of samples $m=50$, and the noise level (a) $\sigma_{n}=0.2\cdot |A|\|\widehat{s}\|_{2}\sqrt{m/|\Omega|}$,
(b) $\sigma_{n}=0.5\cdot |A|\|\widehat{s}\|_{2}\sqrt{m/|\Omega|}$,
and (c) $\sigma_{n}=|A|\|\widehat{s}\|_{2}\sqrt{m/|\Omega|}$. Overall,
the time-of-arrival estimation is reliable in the first case, tenuous
in the second case, and completely unreliable in the third case.}}
\label{fig:Rssnoise}
\end{figure}

\begin{figure}[t]
 \center\includegraphics[height=1.4in]{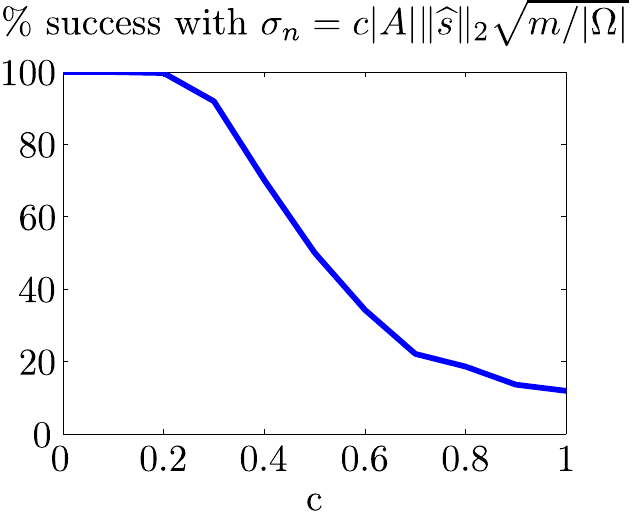} \vspace{-2mm}\caption{\sl\small{Percentage of correct time-of-arrival estimation ($|\widehat{\tau}_{0}-\tau_{0}|\le2a$)
over 1000 trials, as a function of the noise level $\sigma_{n}=c\cdot |A|\|\widehat{s}\|_{2}\sqrt{m/|\Omega|}$,
$0\le c\le1$.}}
\label{fig:Pecentage-of-correct}
\end{figure}

It is worth recalling that a user may have some control over selecting the observation interval $\Omega$. In most cases it would be natural to choose $\Omega$ roughly equal to the essential bandwidth of $s_0$. Taking $\Omega$ larger than this will increase $\mu_1$, $\mu_2$, and the sensitivity to noise; taking $\Omega$ smaller than this will generally increase the width of the main lobe of the autocorrelation function $R_{ss}(\tau)$ and thus limit the resolution to which $\tau_0$ can be estimated.

Finally, we can compare the noise levels in \eqref{eq:sigmarough} and \eqref{eq:sigprecise}
to the noise levels at which a digital matched filter working from
a set of samples taken in the time domain at the Nyquist rate $|\Omega|/2\pi$
will stop being effective. Suppose we sample the (bandlimited) return
signal $As(t-\tau_{0})$ at the Nyquist rate, and noise is added to
these samples. We observe \[
y_{d}[\ell]=As_{d,\tau_{0}}[\ell]+n_{d}[\ell],\quad\text{where}\quad s_{d,\tau_{0}}[\ell]=\left.s(t-\tau_{0})\right|_{t=\ell2\pi/|\Omega|},\]
 and $n_{d}[\ell]$ is sequence of independent zero-mean Gaussian
random variables with variances $\sigma_{n}^{2}|\Omega|/2\pi$. This
variance is chosen to make the noise process similar to that analyzed
in the compressive case; it corresponds to samples of a continuous-time
process that has a power spectral density equal to $\sigma_{n}^{2}$
on $\Omega$ and zero elsewhere.

Focusing just on the complex case for the sake of brevity, once we
have collected $y_{d}$, we can estimate the scaled, shifted autocorrelation
function using \begin{equation}
\widetilde{R}_{d}(\tau)=\<y_{d},s_{d,\tau}\>=A\<s_{d,\tau_{0}},s_{d,\tau}\>+\<n_{d},s_{d,\tau}\>\label{eq:Rnyq}\end{equation}
 and choose as our estimate of $\tau_{0}$ the maximizer of $|\widetilde{R}_{d}(\tau)|$
over all $\tau$. At the correct shift $\tau_{0}$, the first inner
product in \eqref{eq:Rnyq} is given by
\begin{equation*}
\<s_{d,\tau_{0}},s_{d,\tau_{0}}\>=\sum_{\ell}|s_{d,\tau_{0}}[\ell]|^{2}=\frac{|\Omega|}{2\pi}\|s(t-\tau_{0})\|_{2}^{2}=\frac{|\Omega|}{4\pi^{2}}\|\widehat{s}\|_{2}^{2},
\end{equation*}
 where the second equality comes from the fact that $s(\cdot)$ is
bandlimited and we are sampling at the Nyquist rate. The second inner
product is a Gaussian random variable with
\begin{equation*}
\operatorname{Var}\left[\<n_{d},s_{d,\tau}\>\right]~=~\frac{\sigma_{n}^{2}|\Omega|}{2\pi}\sum_{\ell}|s_{d,\tau}[\ell]|^{2}~=~\frac{\sigma_{n}^{2}|\Omega|^{2}}{8\pi^{3}}\|\widehat{s}\|_{2}^{2},
\end{equation*}
 and so
\[
\E|\<n_{d},s_{d,\tau}\>|~\leq~\frac{\sigma_{n}|\Omega|}{2\pi\sqrt{2\pi}}\cdot\|\widehat{s}\|_{2}.\]
 Roughly speaking, then, the Nyquist sampled matched filter will be
overwhelmed by the noise when \begin{equation}
\sigma_{n}~\sim~\left|A\right|\|\widehat{s}\|_{2}.\label{eq:sigmaroughnyq}\end{equation}

Comparing \eqref{eq:sigmaroughnyq} to the compressive matched filter
results \eqref{eq:sigmarough} and \eqref{eq:sigprecise}, we can
interpret the factor of $\sqrt{m/|\Omega|}$ as a sort of {\em undersampling}
penalty; as the number of samples gets smaller, the noise tolerance
gets worse. When $m\gtrsim|\Omega|$, the performance of the two
schemes will be similar. (A similar undersampling penalty arises in standard CS~\cite{davenport2011pros}, where the noise variance that can be tolerated for a given recovery error decreases as the number of measurements gets smaller.)

\section{\label{sec:Pure-Tone-Estimation}Pure Tone Estimation}

As discussed in Section~\ref{sec:exchtf}, the roles of time and frequency are completely interchangeable
in our settings. This allows us to apply the compressive matched filter
to the problem of finding the carrier frequency of a modulated signal
from time-domain samples with minimal effort. One particular such
application is studied here: estimating the frequency of a pure tone
from random samples in time.

Formally the problem under study is described as follows. A pure exponential
$Ae^{\iunit\omega_{0}t}$ with fixed---but unknown---frequency $\omega_{0}\in\Omega$,
amplitude $|A|$, and phase $\measuredangle A$ is observed on $T=[-t_{\mathrm{max}},t_{\mathrm{max}}]$.
Let $y$ be the vector of observations at sampling times $t_{1},t_{2},...,t_{m}\in T$,
which are randomly chosen from a uniform distribution on $T$, i.e.,
\[
y=A\left[\begin{array}{c}
e^{\iunit\omega_{0}t_{1}}\\
e^{\iunit\omega_{0}t_{2}}\\
\vdots\\
e^{\iunit\omega_{0}t_{m}}\end{array}\right].\]

Given $y\in\C^{m}$, we are interested in estimating $\omega_{0}\in\Omega$
and $A\in\C$. A natural approach to solving this problem is to find
the find $\omega_{0}$ and $A$ which best explain the measurements
in a least-squares sense. More formally, we define \begin{equation}
(\widehat{\omega}_{0},\widehat{A}):=\arg\min_{\omega,A}\sum_{k=1}^{m}\left|y[k]-A\cdot e^{\iunit\omega t_{k}}\right|^{2}=\arg\min_{\omega,A}\|y-A\psi_{\omega}\|_{2}^{2},\label{eq:lsmaintone}\end{equation}
 where for any $\omega\in\Omega$, the test vector $\psi_{\omega}\in\mathbb{C}^{m}$
is given by: \[
\psi_{\omega}[k]=e^{\iunit\omega t_{k}},\quad k=1,2,\ldots,m.\]
 The least-squares solution for $\omega_{0}$ is given by \begin{equation}
\widehat{\omega}_{0}=\arg\max_{\omega\in\Omega}\left|\langle y,\psi_{\omega}\rangle\right|,\label{eq:tonels}\end{equation}
 and subsequent to estimating $\omega_{0}$, the least-squares estimate
for $A$ can be computed as $\widehat{A}=\langle y,\psi_{\widehat{\omega}_{0}}\rangle\|\psi_{\widehat{\omega}_{0}}\|_{2}^{-2}$.

\subsection{Analytical Framework}

Equation (\ref{eq:tonels}) suggests a correlation-based strategy
for estimating the unknown frequency $\omega_{0}$. In order to study
the performance of such an estimator, let us define the random process
$X(\omega):=\left\langle y,\psi_{\omega}\right\rangle $ on $\Omega$,
which has the mean function \begin{align*}
\E X(\omega) & =A\E\sum_{k=1}^{m}e^{\iunit(\omega_{0}-\omega)t_{m}}\\
 & =A\sum_{k=1}^{m}\E e^{\iunit(\omega_{0}-\omega)t_{m}}\\
 & =mA\left|T\right|^{-1}\int_{T}e^{\iunit(\omega_{0}-\omega)t}~dt\\
 & =mA|T|^{-1}\cdot|T|\mbox{sinc}\left(\frac{1}{2}\left|T\right|(\omega_{0}-\omega)\right),\end{align*}
 where $\mbox{sinc}\left(\alpha\right):=\sin\left(\alpha\right)/\alpha$.
One way to interpret $X(\omega)$ is that we have approximated the continuous-time
inner product between two time-limited complex sinusoids as a discrete
sum with samples taken at random locations; the above tells us that
this estimate is unbiased.

Further attention reveals that we are facing the same problem as in
the complex case of Section~\ref{sec:Matched filter noise free},
where the roles of time and frequency have been interchanged: the
frequency domain becomes the {}``shift domain,'' while the time
domain becomes the {}``observation domain.'' More precisely, we
may define $\widehat{s}_{0}(\omega)=2\pi\delta(\omega)$ which has
the inverse Fourier transform $s_{0}(t)=1$. Our received signal can
be expressed in the frequency domain as $A\cdot\widehat{s}_{0}(\omega-\omega_{0})$
for some $\omega_{0}\in\Omega$. However, we will observe $m$ samples
of this signal in the time domain, acquiring values of $Ae^{\iunit\omega_{0}t}s_{0}(t)=Ae^{\iunit\omega_{0}t}$
at times $t_{1},t_{2},\dots,t_{m}\in T$.

Now, in the observation (time) domain, we define $s(t)$ to be the
time-limited version of $s_{0}(t)$, i.e., $s(t):=\mathbb{I}_{t\in T}$
where $\mathbb{I}$ denotes the indicator function. Returning to the
shift (frequency) domain, we have $\widehat{s}(\omega)=|T|\mbox{sinc}\left(\frac{1}{2}\left|T\right|\omega\right)$.
Up to a constant factor, this expression equals its own autocorrelation function, i.e., $R_{\widehat{s}\widehat{s}}(\omega)=2\pi\widehat{s}(\omega)$.

Therefore, we can estimate the ideal autocorrelation function $R_{\widehat{s}\widehat{s}}(\cdot)$
(up to the unknown complex amplitude $A$ and translation $\omega_{0}$)
by rescaling the random process $X(\omega)$: \begin{equation}
\widetilde{R}_{\widehat{s}\widehat{s}}(\omega):=\frac{2\pi|T|}{m}X(\omega).\label{eq:acfesttone}\end{equation}
 This estimate is unbiased since $\mathbb{E}[\widetilde{R}_{\widehat{s}\widehat{s}}(\omega)]=A\cdot R_{\widehat{s}\widehat{s}}(\omega-\omega_{0})$.
It is clear that solving the least-squares problem (\ref{eq:tonels})
is equivalent to searching for the maximizer of $|\widetilde{R}_{\widehat{s}\widehat{s}}(\omega)|$.
Since the the main lobe of $R_{\widehat{s}\widehat{s}}(\omega)=2\pi|T|\mbox{sinc}\left(\frac{1}{2}\left|T\right|\omega\right)$
is centered at the origin (with $R_{\widehat{s}\widehat{s}}(0)=2\pi|T|$),
we informally expect that, on average, finding the maximum of $|\widetilde{R}_{\widehat{s}\widehat{s}}(\omega)|$
correctly estimates $\omega_{0}$.

\subsection{Noiseless Analysis}

To study the concentration of $\widetilde{R}_{\widehat{s}\widehat{s}}(\omega)$
about its mean, we may follow the same arguments as in Sections~\ref{sub:mainclean}
and~\ref{sec:gnpc} while simply exchanging the roles of time and
frequency. In particular, we note that the problem of pure tone estimation corresponds to the first {}``special case'' studied in
Section~\ref{sec:gnpc}, because the windowed signal template has
uniform magnitude in the observation domain. Thus, we have $\mu_{1}=\mu_{2}=1$.
This leads us to the following result for the case of noiseless observations.
\begin{corl} \label{corl:Pmaintone}  Fix $\delta > 0$ and let
$$
U = 2\pi C_1 |A||T|\cdot \max\left(\frac{1}{\sqrt{m}},~\frac{\sqrt{\log(4/\delta)}}{m}\right)\cdot\sqrt{\log(12\TO/\delta)}.
$$
If $\left|\Omega\right|\left|T\right|\ge3$, then the estimate of the autocorrelation function in \eqref{eq:acfesttone} obeys
$$
 \Pr\left\{ \sup_{\omega\in\Omega}\left|\widetilde{R}_{\widehat{s}\widehat{s}}(\omega)-AR_{\widehat{s}\widehat{s}}(\omega-\omega_{0})\right|> U\right\} \le \delta.
$$
\end{corl}
This corollary follows immediately from Theorem~\ref{th:Pmain}.

A close inspection of the definition of $\widetilde{R}_{\widehat{s}\widehat{s}}(\omega)$
reveals that $\omega_{0}$ is guaranteed to be a maximizer of $|\widetilde{R}_{\widehat{s}\widehat{s}}(\omega)|$,
with $|\widetilde{R}_{\widehat{s}\widehat{s}}(\omega_{0})|=2\pi|A||T|$.
What Corollary~\ref{corl:Pmaintone} ensures is that even for small values
of $m$, no other values of $\omega$ far from $\omega_{0}$ can give
$|\widetilde{R}_{\widehat{s}\widehat{s}}(\omega)|$ equal (or even
close to) $|\widetilde{R}_{\widehat{s}\widehat{s}}(\omega_{0})|$.
This fact is not only useful when we introduce nonidealities into the observation
process (see Section~\ref{sec:tonerobust}) but also in guiding a
computational method to search for the peak of $|\widetilde{R}_{\widehat{s}\widehat{s}}(\omega)|$.
We investigate this issue in Section~\ref{sec:grid} below.

\subsection{A Grid Search Approach}

\label{sec:grid}

In practice, in order to find the peak of $|\widetilde{R}_{\widehat{s}\widehat{s}}(\omega)|$,
one might hope to simply sample this function over a uniformly spaced
grid of frequencies drawn from $\Omega$. Because $\widetilde{R}_{\widehat{s}\widehat{s}}(\omega)$
is guaranteed to remain close to $AR_{\widehat{s}\widehat{s}}(\omega-\omega_{0})$,
which decays sharply away from $\omega_{0}$, it is possible to guarantee
that as long as the grid is chosen sufficiently fine, then the empirical
maximum over the grid points will occur very close to the true peak.

To illustrate this fact with some specific but arbitrary values, let us note that
for $|\omega|\le\pi|T|^{-1}$, $|R_{\widehat{s}\widehat{s}}(\omega)|\ge 0.636\cdot 2\pi|T|$.
Moreover, for $|\omega|\ge2\pi|T|^{-1}$, $|R_{\widehat{s}\widehat{s}}(\omega)|\le0.218\cdot 2\pi|T|$.
Following the techniques used to prove Corollary~\ref{corl:main},
we can ensure that $0.636\cdot 2\pi\left|A\right||T|- U >0.218\cdot 2\pi\left|A\right||T|+U$
with probability at least $1-\delta$ by taking
\[
m \ge C_2 \cdot \max\left(\frac{\log(12\TO/\delta)}{(0.636-0.218)^{2}},\;\frac{\sqrt{\log(4/\delta)\log(12\TO/\delta)}}{(0.636-0.218)}\right).
\]
It follows that, if we initially search for the maximum of $|R_{\widehat{s}\widehat{s}}(\omega)|$
on a grid with resolution $2\pi|T|^{-1}$ (note that this is the so-called
grid of Nyquist frequencies, given $T$), we are guaranteed that the
empirical maximum will occur at a grid point $\widehat{\omega}_{0}$
such that $|\widehat{\omega}_{0}-\omega_{0}|<2\pi|T|^{-1}$.

After this initial grid search, it is actually straightforward to refine the accuracy of the estimate $\widehat{\omega}_{0}$ using a local concave
ascent. Note that
\begin{align*}
\left|\widetilde{R}_{\widehat{s}\widehat{s}}(\omega)\right|^{2} & =\left(\frac{2\pi|T|}{m}\right)^2\left|\left\langle y,\psi_{\omega}\right\rangle \right|^{2}\\
 & =\left(\frac{2\pi|T||A|}{m}\right)^2\sum_{i,j}e^{\iunit\left(\omega_{0}-\omega\right)\left(t_{i}-t_{j}\right)}\\
 & =\frac{\left(2\pi|T||A|\right)^2}{m}+\left(\frac{2\pi|T||A|}{m}\right)^2\sum_{i\ne j}\cos\left(\left(\omega_{0}-\omega\right)\left(t_{i}-t_{j}\right)\right).\end{align*}
Since $|t_i-t_j|\le|T|$, $|\widetilde{R}_{\widehat{s}\widehat{s}}(\omega)|^{2}$
is guaranteed to be a concave function of $\omega$ when $|\omega-\omega_{0}|\le\frac{\pi}{2}\cdot|T|^{-1}$.
Therefore, if we have an estimate $\widehat{\omega}_{0}$ sufficiently
close to the true $\omega_{0}$ (i.e., $|\widehat{\omega}_{0}-\omega_{0}|\le\frac{\pi}{2}\cdot|T|^{-1}$),
a standard concave maximization (akin to convex minimization) procedure
beginning at $|\widetilde{R}_{\widehat{s}\widehat{s}}(\widehat{\omega}_{0})|$
will give us the {\em exact} value for $\omega_{0}$.
Since the grid search above guarantees that $|\widehat{\omega}_{0}-\omega_{0}|<2\pi|T|^{-1}$, one could ensure success by running four concave maximizations starting from the points $\widehat{\omega}_{0} \pm \frac{\pi}{2}\cdot|T|^{-1}$ and $\widehat{\omega}_{0} \pm \frac{3\pi}{2}\cdot|T|^{-1}$.

\subsection{Robustness}

\label{sec:tonerobust}

It is also possible to consider nonidealities in the observation process.
Following the same set of arguments as in Section~\ref{sec:Robustness}
(but exchanging the roles of time and frequency), we arrive at the
following result.

\begin{corl} \label{corl:tonenoise} Let $N(\omega)$ denote the random process induced by additive complex-valued Gaussian measurement noise having
variance $\sigma_{n}^{2}$, and define
\[
\widetilde{R}_{\widehat{s}\widehat{s},n}(\omega):=\frac{2\pi|T|}{m}(X(\omega)+N(\omega))
\]
to be the estimate of the autocorrelation function formed using the
noisy samples. Let $\delta>0$. Suppose that $\left|\Omega\right|\left|T\right|\ge3$,
that $m\geq C_{3}\log(1/\delta)$, and that
$$
m\ge  C_5 \cdot \max\left( \frac{\log(12\TO/\delta)}{(1-0.218)^2} , \; \frac{\sqrt{\log(4/\delta)\log(12\TO/\delta)}}{1-0.218},\; \frac{ \max\left(\log(\TO),\log(2/\delta)\right)}{(1-0.218)^{2}}\cdot\frac{\sigma_{n}^{2}}{|A|^{2}}  \right).
$$
Then with probability at least $1-2\delta$, the maximum value of $|\widetilde{R}_{\widehat{s}\widehat{s},n}(\omega)|$
must be attained for some $\widehat{\omega}_{0}$ within the interval $[\omega_{0}-2\pi|T|^{-1},\omega_{0}+2\pi|T|^{-1}]$.
\end{corl}

Finally, let us note that with some additional work, we believe it would be possible to extend our analysis to account for multiple tones (or multiple translated pulses in the context of Section~\ref{sec:Matched filter noise free}). The problem becomes that of detecting the true peaks in a noisy sum of $\mathrm{sinc}$ functions. For tones that are well-separated, one could argue that the interference in the random process is minimal and that any prominent peak in $|\widetilde{R}_{\widehat{s}\widehat{s},n}(\omega)|$ indicates the presence of a tone. Tones that are very close may be impossible to discriminate (this is true even with Nyquist-rate samples), while tones that are moderately separated may be possible to discriminate by employing a greedy, iterative estimation procedure.

\subsection{Stylized Application: Chirp Time-of-Arrival Estimation}
\label{subsec:chirp}

We close by noting that the ability to estimate a pure tone's frequency from random time samples can also be parlayed into a technique for estimating a chirp signal's time-of-arrival from random time samples. For this discussion, suppose we receive a chirp signal
$$
x(t) = A \exp\left(j\left(\omega_c(t-t_0) + \frac{\alpha}{2}(t-t_0)^2\right)\right)
$$
over some time interval, where $\omega_c$ denotes the known starting frequency, $\alpha$ denotes the known chirp rate, $A$ denotes the complex amplitude, and $t_0$ denotes the unknown time-of-arrival. We can ``de-chirp'' this signal over this interval, computing
$$
\widetilde{x}(t) = x(t) \exp\left(-j \left(\omega_c t + \frac{\alpha}{2}t^2\right)\right) = \tilde{A}e^{-j\alpha t_0  t},
$$
where $\tilde{A}$ is a complex amplitude. The signal $\widetilde{x}(t)$ is merely a complex sinusoid (in this case, with frequency $\alpha t_0$). We have argued in this section that it is possible to estimate a pure tone's frequency from random samples in time, and in this case that means that it is possible to estimate the time-of-arrival parameter $t_0$ from random samples of $\widetilde{x}(t)$ in time. It is important to note that time samples of $\widetilde{x}(t)$ can be computed easily from time samples of $x(t)$ itself, since the two signals are related via point-wise multiplication.


\section{Theory}
\label{sec:theory}


\subsection{Proofs of Theorems~\ref{th:Emain} and~\ref{th:Pmain} (Noiseless
Analysis)}
\label{sec:mainproofs}

Let us begin by noting that, in the real case, both Theorems~\ref{th:Emain}
and~\ref{th:Pmain} are concerned with bounding
\begin{subequations}
\begin{eqnarray}
\left|\widetilde{R}_{ss}(\tau)-AR_{ss}(\tau-\tau_{0})\right| & = & \left|\widetilde{R}_{ss}(\tau)-A\cdot\mathrm{Re}(R_{ss}(\tau-\tau_{0}))\right|\nonumber \\
 & = & \left|\frac{|\Omega|}{2\pi m}\cdot\mathrm{Re}(X(\tau))-\frac{|\Omega|}{2\pi m}\cdot\mathrm{Re}(\E X(\tau))\right|\nonumber \\
 & = & \frac{|\Omega|}{2\pi m}\left|\mathrm{Re}\left(X(\tau)-\E X(\tau)\right)\right|\nonumber \\
 & \le & \frac{|\Omega|}{2\pi m}\left|X(\tau)-\E X(\tau)\right|.\label{eq:rssxreal}
\end{eqnarray}
In the complex case, both theorems are concerned with bounding
\begin{eqnarray}
\left|\widetilde{R}_{ss}(\tau)-AR_{ss}(\tau-\tau_{0})\right| & = & \frac{|\Omega|}{2\pi m}\left|X(\tau)-\E X(\tau)\right|.\label{eq:rssxcpx}
\end{eqnarray}
\end{subequations}
Thus, to cover both cases, it suffices to focus on bounding $\left|X(\tau)-\E X(\tau)\right|$.

\subsubsection{Setup}

The first step in our approach to bounding $\left|X(\tau)-\E X(\tau)\right|$ is to define the centered process
\[
Y(\tau):=X(\tau)-\E X(\tau)=A\sum_{k=1}^{m}\left|\widehat{s}\left(\omega_{k}\right)\right|^{2}e^{\iunit\omega_{k}\left(\tau-\tau_{0}\right)}-2\pi Am|\Omega|^{-1}R_{ss}(\tau-\tau_{0}).
\]
Our goal is to bound $\sup_{\tau}|Y(\tau)|$, but to do this, we relate the random process to one that is more easily bounded.
First, we symmetrize $Y(\tau)$ in the standard way. Create an independent
copy $Y'(\tau)$ (generated from an independent set of samples $\omega_{1}',\omega_{2}'\ldots,\omega_{m}'$),
and define
\begin{align}
	Z(\tau) & :=Y(\tau)-Y'(\tau)\nonumber \\
	& =A\sum_{k=1}^{m}|\widehat{s}(\omega_{k})|^{2}e^{\iunit\omega_{k}(\tau-\tau_{0})}-|\widehat{s}(\omega_{k}')|^{2}e^{\iunit\omega_{k}'(\tau-\tau_{0})}.\label{eq:Z}
\end{align}
Each term in \eqref{eq:Z} is a symmetric random variable, and so
$Z(\tau)$ has the same distribution as
\[
Z'(\tau):=A\sum_{k=1}^{m}\epsilon_{k}\left(\left|\widehat{s}(\omega_{k})\right|^{2}e^{\iunit\omega_{k}(\tau-\tau_{0})}-\left|\widehat{s}(\omega_{k}')\right|^{2}e^{\iunit\omega_{k}'(\tau-\tau_{0})}\right),
\]
where $\epsilon_{1},\epsilon_{2},...,\epsilon_{m}$ is a Rademacher
sequence independent of everything.%
\footnote{A Rademacher sequence is a sequence of independent random variables
taking $\pm1$ values with equal probabilities.%
}

We can control $\E\sup_\tau|Y(\tau)|$ through $\E\sup_\tau|Z'(\tau)|$ using the following simple result, which is proved in Appendix~\ref{sec:Y to Z}.
\begin{lem} \label{lem:Y to Z}
$\E\sup_{\tau}|Y(\tau)|\le\E\sup_{\tau}|Z'(\tau)|$.
\end{lem}
Furthermore, the deviation of $\sup_\tau |Y(\tau)|$ from its average can be controlled through the corresponding deviation of $\sup_\tau |Z'(\tau)|$.
The following is proved in Appendix~\ref{sec:Y to Z2}.
\begin{lem} \label{lem:Y to Z2} For any $\lambda\ge0$,
\[
	\P{\sup_{\tau}|Y(\tau)| > 2\E\sup_\tau|Y(\tau)| + \lambda} ~\leq~
	2\P{\sup_{\tau}|Z'(\tau)| > \lambda}.
\]
\end{lem}
The above results allow us to focus on developing expectation and tail bounds for $\sup_{\tau}|Z'(\tau)|$.
We establish such bounds in the following subsections.

To ease the notation below, we make the following definitions for
quantities that will appear often: \[
M=M(s):=\sup_{\omega\in\Omega}|\widehat{s}(\omega)|^{2}=\|\widehat{s}\|_{\infty}^{2},\]
 \[
\bigmone=\bigmone(s,m,\Omega):=\sqrt{\frac{m}{|\Omega|}}\|\widehat{s}\|_{4}^{2},\]
 \[
\bigmtwo=\bigmtwo(s,m,\Omega):=\sqrt{\frac{m}{|\Omega|}}\|\widehat{s}\|_{2},\]
 \[
\bigmthree=\bigmthree(s,\omega_{1},\omega_{2},\dots,\omega_{m},\omega_{1}',\omega_{2}',\dots,\omega_{m}'):=\sqrt{\sum_{k=1}^{m}|\widehat{s}(\omega_{k})|^{4}+|\widehat{s}(\omega_{k}')|^{4}},\]
 and \[
\bigmfour=\bigmfour(s,\omega_{1},\omega_{2},\dots,\omega_{m}):=\sqrt{\sum_{k=1}^{m}|\widehat{s}(\omega_{k})|^{2}}.\]
We will also frequently use the following convenient facts. For any $a$ and $b$, we have
\begin{equation}
	e^{\iunit a}-e^{\iunit b}=
	2\iunit\sin\left(\frac{a-b}{2}\right)e^{\iunit(a+b)/2},
	\label{eq:conv fact 1}
\end{equation}
and
\begin{equation}
	\left|a+b\right|^{2}\le2\left|a\right|^{2}+2\left|b\right|^{2}.
	\label{eq:conv fact 2}
\end{equation}
Also, for any $c,u>0$, the following inequality follows from a standard Gaussian tail bound~\cite{CntUnivarDists}:
\begin{equation}
	\int_{x\ge u}e^{-\frac{x^2}{c^2}} dx\le \frac{c^2}{2u}e^{-\frac{u^2}{c^2}}.
\label{eq:conv fact 3}
\end{equation}

\subsubsection{Chaining}
\label{sec:pfchain}

We start by bounding $\sup_{\tau}\left|Z'(\tau)\right|$ conditioned on the choice of $\{\omega_{k}\}$ and $\{\omega'_{k}\}$.
To this end, we will use a chaining argument similar to what is used to prove the general Dudley inequality \cite{dudley1967sizes}, but optimized for our particular process (this will allow us to tightly control the constants).

The following tail bounds for $Z'(\tau)$ and its increments are proved in Appendix~\ref{sec:inctails}.
\begin{lem}
	\label{lem:inctails}
	For a fixed $\tau\in \R$ and any $\lambda\ge0$, $Z'(\tau)$ obeys
	\begin{equation}
		\operatorname{P}_{\epsilon_k}\left\{|Z'(\tau)| > \lambda\right\} ~\leq~ 2\exp\left(-\frac{\lambda^2}{4|A|^2M_3^2}\right),
    \label{eq:tailinc1}
	\end{equation}
    where $\operatorname{P}_{\epsilon_k}$ denotes probability with respect to $\{\epsilon_k\}$ conditioned on fixed $\{\omega_{k}\}$ and $\{\omega'_{k}\}$. Also, for fixed $\tau_1,\tau_2 \in \R$,
	\begin{equation}
		\operatorname{P}_{\epsilon_k}\left\{|Z'(\tau_1)-Z'(\tau_2)| > \lambda\right\}
		~\leq~ 2\exp\left(-\frac{\lambda^2}{|A|^2M_3^2|\Omega|^2|\tau_1-\tau_2|^2}\right).
	\label{eq:tailinc2}
    \end{equation}
\end{lem}

We will consider the values of $Z'(\tau)$ on a series of discrete grids of points that are essentially localized on the interval $T=[\tau_{\mathrm{min}},\tau_{\mathrm{max}}]$. For each integer $j \ge 0$, let $T_j$ be a grid of points spaced $2^{-j}|\Omega|^{-1}$ apart:
\begin{equation}
	\label{eq:Tj}
	T_j = \{\tau_{\mathrm{min}} + 2^{-j-1}|\Omega|^{-1} + k2^{-j}|\Omega|^{-1},~k = 0,1,\ldots, \lfloor 2^{j}\TO \rfloor \}.
\end{equation}
All points in $T_j$ belong to $T$, except possibly the final point in $T_j$, which may exceed $\tau_{\mathrm{max}}$ by no more than $2^{-j-1}|\Omega|^{-1}$.
Moreover, if we denote by $\pi_j(\tau)$ the closest point in $T_j$ to a given point $\tau$, then $|\tau-\pi_j(\tau)|\leq 2^{-j-1}|\Omega|^{-1}$ for all $\tau\in T$.
The points in the $T_j$ are arranged like nodes in a dyadic tree, with each ``parent'' in $T_j$ having two ``children'' in $T_{j+1}$ (the two points that are closer to the parent than to any other point in $T_{j}$); the only exception to this rule occurs if $|T_{j+1}|$ is odd, in which case the final point in $T_j$ has only one child in $T_{j+1}$.

We define $L_j$ to be the set of ``links'' that connect the parents in $T_j$ to their children in $T_{j+1}$:
\begin{equation}
	\label{eq:Lj}
	L_j = \{(p,q)\in(T_j,T_{j+1})
	~|~\pi_j(\tau)=p~\text{and}~\pi_{j+1}(\tau)=q~\text{for some $\tau\in T$} \}.
\end{equation}
Because of the one-dimensional structure of $T$ and the particular arrangement of $T_j$'s, we observe that every child in $T_{j+1}$ is associated with only one link, and thus  $\#L_j=\#T_{j+1}\le2^{j+1}\TO+1$. Furthermore, the length of every link is half of the distance between consecutive points on $T_{j+1}$; that is $|q_j-p_j|=2^{-j-2}|\Omega|^{-1}$ for all $(p_j, q_j) \in L_j$.

For almost every $\tau\in T$~\cite[(6.46)]{ra09-1}, we can decompose $Z'(\tau)$ as a sum of the differences between approximations at different scales, writing the telescoping sum
\begin{equation*}
	Z'(\tau) = Z'(\pi_0(\tau)) + \sum_{j\geq 0} Z'(\pi_{j+1}(\tau)) - Z'(\pi_j(\tau)).
\end{equation*}
Thus
\[
	|Z'(\tau)| ~\leq~
	|Z'(\pi_0(\tau))| + \sum_{j\geq 0} |Z'(\pi_{j+1}(\tau)) - Z'(\pi_j(\tau))|,
\]
and
\[
	\sup_{\tau\in T}|Z'(\tau)| ~\leq~
	\max_{p_0\in T_0}|Z'(p_0)| +
	\sum_{j\geq 0} \max_{(p_j,q_j)\in L_j}|Z'(q_j) - Z'(p_j)|.
\]	
Therefore, for any $\lambda_1,\lambda_2\ge0$ and any sequence of positive numbers $\{ u_j\}$ such that $\sum_{j\geq0}u_j\leq1$, we have
\begin{align}
	& \PE{\sup_{\tau\in T} |Z'(\tau)| > \lambda_1 + \lambda_2} \nonumber \\
    &~~~~~\leq
	\PE{ \left \{\max_{p_0\in T_0}|Z'(\tau)| > \lambda_1 \right\} \cup \left\{\sum_{j\geq0} \max_{(p_j,q_j)\in L_j}|Z'(q_j)- Z'(p_j)| > \lambda_2\right\}} \nonumber \\
    &~~~~~ \leq\PE{\max_{p_0\in T_0}|Z'(\tau)| > \lambda_1} +
	\PE{\sum_{j\geq0} \max_{(p_j,q_j)\in L_j}|Z'(q_j)- Z'(p_j)| > \lambda_2} \nonumber \\
	&~~~~~\le \PE{\max_{p_0\in T_0}|Z'(\tau)| > \lambda_1} +
	\sum_{j\geq0} \PE{ \max_{(p_j,q_j)\in L_j}|Z'(q_j)- Z'(p_j)| > \lambda_2 u_j}. \label{eq:twoPartSup}
\end{align}
To bound the first term in (\ref{eq:twoPartSup}), we apply (\ref{eq:tailinc1}) along with the union bound and the fact that $\#T_0\leq \TO + 1$ to obtain
$$
	\PE{\max_{p_0\in T_0}|Z'(\tau)| > \lambda_1} \leq 2(\#T_0)\exp\left(\frac{-\lambda_1^2}{4|A|^2M_3^2}\right)
       \leq 2(\TO+1)\exp\left(\frac{-\lambda_1^2}{4|A|^2M_3^2}\right).
$$
To bound the second term in (\ref{eq:twoPartSup}), take $u_j = \sqrt{j+3}~2^{-j-2}$ and assume that $\lambda_2 \ge \lambda_0/3$, where $\lambda_0:=3|A|M_3\sqrt{\log{(2\TO)}}$. 
Then, for every $j\ge0$, we have
\begin{align*}
	\PE{\max_{(p_j,q_j)\in L_j}|Z'(q_j)- Z'(p_j)| > \lambda_2 u_j}
	& \le 2(\#L_j)\exp\left(\frac{-\lambda_2^2 u_j^2}{|A|^2M_3^2|\Omega|^2(2^{-j-2}|\Omega|^{-1})^2}\right)\\
    &= 2(\#L_j)\exp\left(\frac{-(j+2)\lambda_2^2}{|A|^2M_3^2}\right)
	 \exp\left(\frac{-\lambda_2^2}{|A|^2M_3^2}\right)\\
	&\leq 2(\#L_j)(2\TO)^{-j-2}\exp\left(\frac{-\lambda_2^2}{|A|^2M_3^2}\right) \\
	&\le \left((\TO)^{-j-1}+2^{-j-1}(\TO)^{-j-2}\right)\exp\left(\frac{-\lambda_2^2}{|A|^2M_3^2}\right),
\end{align*}
where the first line above follows from applying (\ref{eq:tailinc2}) along with the union bound and the fact that $|q_j-p_j|=2^{-j-2}|\Omega|^{-1}$ for all $(p_j, q_j) \in L_j$, the third line uses the assumption that $\lambda_2 \ge \lambda_0/3$, and the fourth line follows because $\#L_j\leq 2^{j+1}\TO + 1$.
If we assume that $\TO\geq 3$, it follows that
\begin{align*}
	\sum_{j\geq0} \PE{ \max_{(p_j,q_j)\in L_j}|Z'(q_j)- Z'(p_j)| > \lambda_2 u_j} & \le \sum_{j\ge0} \left(3^{-j-1}+2^{-j-1}\cdot3^{-j-2}\right)\exp\left(\frac{-\lambda_2^2}{|A|^2M_3^2}\right)\\
	&\le(17/30)\exp\left(\frac{-\lambda_2^2}{|A|^2M_3^2}\right).
\end{align*}

Putting together our bounds for the first and second terms in (\ref{eq:twoPartSup}), for any $\lambda \ge \lambda_0$ we may take $\lambda_1 = 2\lambda/3$ and $\lambda_2 = \lambda/3$ to conclude that
\begin{align}
	\PE{\sup_{\tau\in T} |Z'(\tau)| > \lambda} &\leq
	(2\TO+2)\exp\left(\frac{-\lambda^2}{9|A|^2M_3^2}\right) +
	(17/30)\exp\left(\frac{-\lambda^2}{9|A|^2M_3^2}\right) \nonumber\\
	&\leq \left(2\TO+2.57\right)\exp\left(\frac{-\lambda^2}{9|A|^2M_3^2}\right)\nonumber\\
	&\leq 3\TO \exp\left(\frac{-\lambda^2}{9|A|^2M_3^2}\right).\label{eq:tail bound noise-free cond on omega_k}	
\end{align}
The third line above follows from our assumption that $\TO\ge3$. In the subsections that follow, we translate this conditional tail bound into unconditional expectation and tail bounds for $\sup_{\tau}|Z'(\tau)|$.

\subsubsection{Completing the Proof of Theorem~\ref{th:Emain} (Expectation)}

\label{sub:expectation noise-free process}

Conditioned on the choice of $\{\omega_{k}\}$ and $\{\omega'_{k}\}$, we can integrate the tail bound developed above to obtain an upper bound for
$\operatorname{E}_{\epsilon_k}\sup_{\tau}\left|Z'(\tau)\right|$. Note that, for any nonnegative random variable $V$, we have~\cite[Prop.~6.1]{ra09-1}
\begin{equation}
\label{eq:tailtoE}
	\E V = \int_0^\infty \P{V > u} du.
\end{equation}

Once we have bounded the average of the supremum of the conditioned process, it is then straightforward to extend this to a bound for $\E\sup_{\tau}\left|Z'(\tau)\right|$
by removing the conditioning on $\{\omega_{k}\}$ and $\{\omega'_{k}\}$.

Recall that $\lambda_0=3|A|M_3\sqrt{\log{(2\TO)}}$. Then it follows from the identity above that
\begin{align}
	\mathbb{E}_{\epsilon_k}\sup_{\tau\in T}|Z'(\tau)|
	&= \int_{0}^\infty \operatorname{P}_{\epsilon_k} \left\{\sup_{\tau\in T} |Z'(\tau)|>\lambda\right\} d\lambda \nonumber\\
	&= \int_{0}^{\lambda_0} \operatorname{P}_{\epsilon_k}\left\{\sup_{\tau\in T} |Z'(\tau)|>\lambda\right\} d\lambda
	+\int_{\lambda_0}^{\infty} \operatorname{P}_{\epsilon_k}\left\{\sup_{\tau\in T} |Z'(\tau)|>\lambda\right\} d\lambda \nonumber\\
	&\leq \lambda_0+\left(2\TO+2.57\right)\int_{\lambda_0}^{\infty} \exp\left(\frac{-\lambda^2}{9|A|^2M_3^2}\right) d\lambda \nonumber\\
	&\leq \lambda_0+\left(2\TO+2.57\right)\cdot \frac{9|A|^2M_3^2}{2\lambda_0}\exp\left(\frac{-\lambda_0^2}{9|A|^2M_3^2}\right) \nonumber\\
	&= 3|A|M_3\sqrt{\log{(2\TO)}}+\left(1+\frac{1.285}{\TO}\right)\frac{3|A|M_3}{2\sqrt{\log(2 \TO)}}\nonumber\\
	&\le 3|A|M_3\sqrt{\log{(2\TO})}+\left(1+\frac{1.285}{3}\right)\frac{3|A|M_3}{2\sqrt{\log6}}\nonumber\\
	&\le |A|M_3\left(3\sqrt{\log{(2\TO)}}+1.61\right).\label{eq:main bnd on E cond on omega_k}
\end{align}
The fourth line above follows from \eqref{eq:conv fact 3}, and the sixth line follows from our assumption that  $\TO\ge3$. Now it remains to remove the conditioning by taking the average over $\left\{ \omega_k \right\}$ and $\left\{ \omega'_k \right\}$. First note that, by Jensen's inequality, we have
\begin{align}
	\mathbb{E}_{\omega_k,\omega'_k}M_3 &\leq \sqrt{\mathbb{E}_{\omega_k,\omega'_k}M_3^2}\nonumber\\
	&=\sqrt{\E_{\omega_k} \sum_{k=1}^m |\widehat{s}(\omega_k)|^4  + \E_{\omega_k'} \sum_{k=1}^m|\widehat{s}(\omega_k')|^4 } \nonumber\\
	&=\sqrt{2m\mathbb{E}_{\omega}|\widehat{s}(\omega)|^4}\nonumber\\
	&=\sqrt{2m|\Omega|^{-1}\left\Vert\widehat{s}\right\Vert^4_4}\nonumber\\
	&=\sqrt{2}M_1.\label{eq:average of M_3}
\end{align}
Now combining the above inequality with \eqref{eq:main bnd on E cond on omega_k} brings us to
\begin{align*}
	\E \sup_{\tau\in T}|Z'(\tau)| &=\operatorname{E}_{\omega_k,\omega'_k}\operatorname{E}_{\epsilon_k} \sup_{\tau\in T}
	 |Z'(\tau)|\\
	&\le \operatorname{E}_{\omega_k,\omega'_k} |A|M_3\left(3\sqrt{\log{(2\TO)}}+1.61\right)\\
	&\le |A|M_1\left(4.25\sqrt{\log{(2\TO)}}+2.28\right).
\end{align*}
The final link to the random process of interest is via Lemma~\ref{lem:Y to Z}:
\begin{align}
	\E\sup_{\tau}|Y(\tau)|&\le\E\sup_{\tau}|Z'(\tau)|\nonumber\\
	&\le |A|M_1\left(4.25\sqrt{\log{(2\TO)}}+2.28\right)\label{eq:bnd on average of the noisefree process Pre}\\
	&\le |A|M_1\left(4.25\sqrt{\log{(2\TO)}}+\frac{2.28}{\sqrt{\log 6}}\sqrt{\log{(2\TO)}}\right)\nonumber\\
	&\le 5.96|A|M_1\sqrt{\log{(2\TO)}}, \label{eq:bnd on average of the noisefree process}
\end{align}
where the third line follows because we assumed that $\TO\ge3$. This completes the proof of Theorem \ref{th:Emain}, after plugging \eqref{eq:bnd on average of the noisefree process Pre} and  \eqref{eq:bnd on average of the noisefree process} into~\eqref{eq:rssxreal} or~\eqref{eq:rssxcpx}.

\subsubsection{Completing the Proof of Theorem~\ref{th:Pmain} (Tail Bound)}

\label{sub:tail}

Recall the tail bound obtained in \eqref{eq:tail bound noise-free cond on omega_k}. In this section, we remove the conditioning on $\{\omega_k\}$ and $\{\omega'_k\}$
to obtain a tail bound for the supremum of $\left|Y(\tau)\right|$ on $T$.
From \eqref{eq:average of M_3}, recall that $M_3^2$ is a sum of independent bounded random variables, and thus it
is closely concentrated about its average $\E M_3^2 = 2M_1^2$. To quantify this, we will use the classic Bernstein inequality, which is restated
below for convenience.
\begin{lem} \label{lem:Bernstein ineq} {\em \cite{bennett1962probability}}
Consider a sequence of independent zero-mean random variables $V_{1},V_{2},\ldots,V_{m}$
with $|V_{k}|\leq B$ for $k=1,2,...,m$. Then for any $\lambda \geq 0$, the following holds:
\[
	\P{ \sum_{k=1}^{m}V_{k}>\lambda} ~\leq~\exp\left(\frac{-\lambda^{2}}{2\rho^{2}+2B\lambda/3}\right),
\]
 where $\rho^{2}=\sum_{k=1}^{m}\mathbb{E}V_{k}^{2}$.
 \end{lem}

Here, take $V_k=|\widehat{s}(\omega_k)|^4+|\widehat{s}(\omega'_k)|^4-2|\Omega|^{-1}\left\Vert\widehat{s}\right\Vert^4_4$ and note that
\begin{align*}
	\E V_k & =2\E \left(|\widehat{s}(\omega_k)|^4-|\Omega|^{-1}\left\Vert\widehat{s}\right\Vert^4_4\right)=0.
\end{align*}
Also,
\begin{align*}
	|V_k| &\leq \sup_{\omega,\omega'}\left| |\widehat{s}(\omega)|^4 + |\widehat{s}(\omega')|^4 - 2|\Omega|^{-1}\left\Vert\widehat{s}\right\Vert^4_4 \right| \\
	&\le \sup_{\omega,\omega'} ~ \max\left( |\widehat{s}(\omega)|^4 + |\widehat{s}(\omega')|^4, 2|\Omega|^{-1}\left\Vert\widehat{s}\right\Vert^4_4\right)\\
	&= \max\left( 2\sup_\omega |\widehat{s}(\omega)|^4 , 2|\Omega|^{-1}\left\Vert\widehat{s}\right\Vert^4_4\right)\\
	&= 2\max \left( M^2 , |\Omega|^{-1}\left\Vert\widehat{s}\right\Vert^4_4\right)\\
	&= 2M^2.
\end{align*}
The second line above follows from the convenient fact that $\left|a-b\right|\le\max\left(a,b\right)$ for any $a,b\ge0$, and the last line follows from the fact that $\left\Vert\widehat{s}\right\Vert_4^4 \le |\Omega|\sup_\omega|\widehat{s}(\omega)|^4$. We also have 
\begin{align*}
	\rho^2 &= \sum_{k=1}^m \E\left(|\widehat{s}(\omega_k)|^4 + |\widehat{s}(\omega'_k)|^4\right)^2 - 4|\Omega|^{-2}\left\Vert\widehat{s}\right\Vert_4^8\\
	&\leq \sum_{k=1}^m \E\left(|\widehat{s}(\omega_k)|^4 + |\widehat{s}(\omega_k')|^4\right)^2 \\
	&\leq \sum_{k=1}^m 2\E|\widehat{s}(\omega_k)|^8 + 2\E\left|\widehat{s}(\omega'_k)\right|^8 \\
	&= 4m\E|\widehat{s}(\omega)|^8 \\
	&\leq 4mM^2 \E|\widehat{s}(\omega)|^4 \\
	&=4mM^2|\Omega|^{-1}\|\widehat{s}\|^4_4\\
	&= 4M^2M_1^2.
\end{align*}
The first line above follows because $\E(V-\E V)^2=\E V^2-(\E V)^2$ for any real-valued random variable $V$. 
The third line is implied by \eqref{eq:conv fact 1}. Now we can apply the Bernstein inequality to $\sum_{k=1}^mV_k$ for any $\lambda\ge0$ and obtain
\begin{equation}
	\P{M_3^2 > 2M_1^2 + \lambda}
	\le \exp\left(\frac{-\lambda^2}{8M^2M_1^2+(4/3)M^2\lambda}\right).\label{eq:final bernstein for M_3^2}
\end{equation}
Assume that $M_1\geq M\sqrt{\log(4/\delta)}$.  Then \eqref{eq:final bernstein for M_3^2} implies that
\[
	\P{M_3^2 > 2M_1^2 + \lambda} ~\leq~ \exp\left(\frac{-\lambda^2\log(4/\delta)}{8M_1^4 + 4M_1^2\lambda/3}\right).
\]
Take $\lambda = 3.58 M_1^2$.  Then
\begin{equation}
	\P{M_3^2 > 5.58M_1^2}\leq \delta/4.\label{eq:M3 tail part 1}
\end{equation}

Now assume that $M_1\leq M\sqrt{\log(4/\delta)}$.  Then
\begin{align*}
	\P{M_3^2 > 2M^2\log(4/\delta) + \lambda}
	&\leq \P{M_3^2 > 2M_1^2 + \lambda}\\
	&\le \exp\left(\frac{-\lambda^2}{8M^4\log(4/\delta) + (4/3)M^2\lambda}\right).
\end{align*}

Take $\lambda = 3.58 M^2\log(4/\delta)$.  Then
\begin{equation}
	\P{M_3^2 > 5.58M^2\log(4/\delta)}
	\leq\delta/4.\label{eq:M3 tail part 2}
\end{equation}

Therefore, combining \eqref{eq:M3 tail part 1} and \eqref{eq:M3 tail part 2}, we arrive at
\begin{equation*}
	\P{M_3 > 2.37 \max\left(M_1,M\sqrt{\log(4/\delta)}\right) } ~\leq~\delta/4.
\end{equation*}
Let $\mathcal{E}$ denote the event that $M_3\le2.37 \max\left(M_1,M\sqrt{\log(4/\delta)}\right)$. Then clearly,
\begin{equation}
\label{eq:tail bnd on noise free part1}
	\P{\mathcal{E}}\ge 1-\delta/4.
\end{equation}

On the other hand, taking $\lambda=3|A|M_3\sqrt{\log(12\TO/\delta)}$, \eqref{eq:tail bound noise-free cond on omega_k} implies that
\begin{equation}\label{eq:tail bnd on noise free part2}
	\PE{\sup_\tau|Z'(\tau)| > 3|A|M_3\sqrt{\log(12\TO/\delta)}} ~\leq~
	\delta/4,
 \end{equation}

Now we can combine \eqref{eq:tail bnd on noise free part1} and \eqref{eq:tail bnd on noise free part2} as follows. For notational convenience, set  $b:=3|A|M_3\sqrt{\log(12\TO/\delta)}$, define
\[
	u := 7.11|A| \max\left(M_1,M\sqrt{\log(4/\delta)}\right)
	\sqrt{\log(12|\Omega||T|/\delta)},
\]
and note that
\begin{align*}
	& \P{ \sup_{\tau\in T}\left|Z'\left(\tau\right)\right|> u} \\
	& \le\P{ \sup_{\tau\in T}\left|Z'\left(\tau\right)\right|> u\,\middle|\, \mathcal{E}} \P{\mathcal{E}} +\P{ \mathcal{E}^C} \\
	& \le\P{ \sup_{\tau\in T}\left|Z'\left(\tau\right)\right|> b\,\middle|\, \mathcal{E}} \P{\mathcal{E}} +\frac{\delta}{4}\\
	& =\P{ \mathcal{E}}\cdot\frac{1}{\P{ \mathcal{E}}}\int_{\mathcal{E}}\PE{ \sup_{\tau\in T}\left|Z'\left(\tau\right)\right|> b\,\middle|\,\left\{ \omega_{k}\right\}, \left\{ \omega_{k}'\right\} } \, d\mu\left(\left\{ \omega_{k}\right\}, \left\{ \omega_{k}'\right\} \right)+\frac{\delta}{4}\\
	& \le\int_{\mathcal{E}}\frac{\delta}{4}\, d\mu\left(\omega_{1},\omega_2,\cdots,\omega_m \right)+\frac{\delta}{4}\\
	& =\frac{\delta}{2}.
\end{align*}
The fourth line above follows from the definition of conditional measure.
The final link to the random process of interest is via Lemma~\ref{lem:Y to Z2}:
\begin{align}
	\P{\sup_{\tau}|Y(\tau)| > 8.5|A|M_1\sqrt{\log(2\TO)}+4.56|A|M_1 + u} &
	\le\P{\sup_{\tau}|Y(\tau)| > 2\E\sup_\tau|Y(\tau)| + u} \nonumber\\
	&\le 2\P{\sup_{\tau}|Z'(\tau)| > u}\nonumber \\
	&\le \delta,\label{eq:tail bnd on noise free process}
\end{align}
where we used \eqref{eq:bnd on average of the noisefree process Pre} in the first line above. This completes the proof of Theorem~\ref{th:Pmain}, as
\begin{align*}
	&8.5|A|M_1\sqrt{\log(2\TO)}+4.56|A|M_1 + u\\
	&\le  |A|\max\left(M_1,M\sqrt{\log(4/\delta)}\right) \left(8.5\sqrt{\log(2\TO)}+4.56+7.11\sqrt{\log(12|\Omega||T|/\delta)}\right)\\
	&\le  |A|\max\left(M_1,M\sqrt{\log(4/\delta)}\right) \left(15.61\sqrt{\log(12|\Omega||T|/\delta)}+4.56\right)\\
	&\le  |A|\max\left(M_1,M\sqrt{\log(4/\delta)}\right) \left(15.61\sqrt{\log(12|\Omega||T|/\delta)}+\frac{4.56}{\sqrt{\log 36}}\sqrt{\log(12|\Omega||T|/\delta)} \right)\\
	&\le  C_1|A|\max\left(M_1,M\sqrt{\log(4/\delta)}\right) 	\sqrt{\log(12|\Omega||T|/\delta)},
\end{align*}
where the fourth line follows from the assumption that $\TO\ge3$, and the fifth line holds by taking $C_1 = 18.02$. Actually, we see from the above that we may slightly improve upon the value of $U$ specified in \eqref{eq:U} by taking
	\begin{equation*}
		U=\max\left(\frac{\eta\mu_{1}}{\sqrt{m}},~~\frac{\eta\mu_{2}}{m}\cdot\sqrt{\log(4/\delta)}\right)\cdot \left(15.61\sqrt{\log(12|\Omega||T|/\delta)}+4.56\right).
	\end{equation*}


\subsection{Proofs of Theorems~\ref{thm:noiseexp} and~\ref{thm:mainNoisenew}
(Noisy Measurements)}
\label{sec:noiseproofs}

Let us begin by noting that, in the real case, both Theorems~\ref{thm:noiseexp}
and~\ref{thm:mainNoisenew} are concerned with bounding $|\widetilde{N}(\tau)|=\frac{|\Omega|}{2\pi m}\left|\mathrm{Re}(N(\tau))\right|\le\frac{|\Omega|}{2\pi m}\left|N(\tau)\right|$.
In the complex case, both theorems are concerned with bounding $|\widetilde{N}(\tau)|=\frac{|\Omega|}{2\pi m}\left|N(\tau)\right|$.
Thus, to cover both cases, it suffices to focus on bounding $\left|N(\tau)\right|$.

We first bound $\E\sup_{\tau}\left|N(\tau)\right|$, and then we show that
$\sup_{\tau}\left|N(\tau)\right|$ is sharply concentrated about its
mean with high probability.

\subsubsection{Proof of Theorem~\ref{thm:noiseexp} (Expectation)
\label{sub:Dudley-Inequality for add noise}}

We begin by noting that conditioned on $\{\omega_{k}\}$,
$N(\tau)$ is a complex-valued Gaussian process to which we can apply a chaining argument similar to the one put forth in Section~\ref{sec:mainproofs}.  With the $\{\omega_k\}$ fixed, at each $\tau$ the real and imaginary parts of $N(\tau)$ have the same Gaussian distributions, and thus the magnitude $|N(\tau)|$ is a Rayleigh random variable with second moment
\begin{align}
	\E_{n_k}|N(\tau)|^2 & =\E_{n_k} \left|\sum_{k=1}^m n_k\widehat{s}^*(\omega_k)e^{\iunit \omega_k\tau}\right|^2\nonumber\\
	 &=\sum_{k=1}^m\E_{n_k}|n_k|^2 \left| \widehat{s}(\omega_k) \right|^2\nonumber\\
	&= \sigma_n^2\sum_{k=1}^m|\widehat{s}(\omega_k)|^2 \nonumber\\
	&=\sigma_n^2M_4^2,\label{eq:2nd moment of |N(t)|}
\end{align}
where the second line follows from the independence of $\{n_k\}$.
It is known that a Rayleigh random variable $V$ with $\E V^2=c^2$ satisfies $\P{V>\lambda}=\exp(-\lambda^2/c^2)$ \cite{CntUnivarDists}, and thus a tail bound for $|N(\tau)|$ follows directly from the above:
\begin{equation}
	\label{eq:Ntautail}
	\operatorname{P}_{n_k}\left\{|N(\tau)| > \lambda \right\} ~=~
	\exp\left(\frac{-\lambda^2}{\sigma_n^2M_4^2}\right).
\end{equation}
Likewise, the increment $|N(\tau_1)-N(\tau_2)|$ is Rayleigh with
\begin{align*}
	\E_{n_k}|N(\tau_1)-N(\tau_2)|^2 &=
	\E_{n_k} \left|\sum_{k=1}^m n_k\widehat{s}^*(\omega_k)\left(e^{\iunit \omega_k\tau_1}-e^{\iunit\omega_k\tau_2}\right) \right|^2\\
	 &=\sum_{k=1}^m \E_{n_k} |n_k|^2 \left|\widehat{s}(\omega_k)\right|^2 \left|e^{\iunit \omega_k\tau_1}-e^{\iunit\omega_k\tau_2}\right|^2\\
	&=\sigma_n^2\sum_{k=1}^m|\widehat{s}(\omega_k)|^2\left|e^{\iunit\omega_k\tau_1}-e^{\iunit\omega_k\tau_2}\right|^2 \\
	&= \sigma_n^2\sum_{k=1}^m4|\widehat{s}(\omega_k)|^2 \sin^2\left(\omega_k(\tau_1-\tau_2)/2\right) \\
	&\le \sigma_n^2\sum_{k=1}^m|\widehat{s}(\omega_k)|^2 |\omega_k|^2 |\tau_1-\tau_2|^2 \\
	&\le \frac{1}{4}\sigma_n^2|\Omega|^2|\tau_1-\tau_2|^2\sum_{k=1}^m|\widehat{s}(\omega_k)|^2   \\
	&= \frac{1}{4}\sigma_n^2M_4^2|\Omega|^2|\tau_1-\tau_2|^2.
\end{align*}
The second line above follows from the independence of $\{n_k\}$, and  the fourth line follows from \eqref{eq:conv fact 1}. The tail bound for the increment is then
\begin{equation}
	\label{eq:Ninctail}
	\operatorname{P}_{n_k}\left(|N(\tau_1)-N(\tau_2)| > \lambda\right) ~\leq~
	\exp\left(\frac{-4\lambda^2}{\sigma_n^2M_4^2|\Omega|^2|\tau_1-\tau_2|^2}\right).
\end{equation}

With the same definition of the sets $T_j$ and $L_j$ from \eqref{eq:Tj} and \eqref{eq:Lj} in Section~\ref{sec:pfchain}, we can write the telescoping sum for $N(\tau)$ and proceed similarly to obtain
\[
 	\sup_{\tau\in T}|N(\tau)| ~\leq~
	\max_{p_0\in T_0}|N(p_0)| + \sum_{j\geq 0}\max_{(p_j,q_j)\in L_j} |N(q_j)-N(p_j)|,
\]
and so it follows immediately that
\begin{equation}
 	\E_{n_k}\sup_{\tau\in T}|N(\tau)| ~\leq~
	\E_{n_k}\max_{p_0\in T_0}|N(p_0)| + \sum_{j\geq 0}\E_{n_k}\max_{(p_j,q_j)\in L_j} |N(q_j)-N(p_j)|.
\label{eq:telescope for N}
\end{equation}

We now use the following standard result that bounds the expected maximum of a finite set of subgaussian random variables. The proof is included in Appendix \ref{sec:Proof of Emaxsubg}.
\begin{prop}
	\label{prop:Emaxsubg}
	Let $V_1,V_2,\ldots,V_N$ be random variables with $\P{|V_i| > \lambda} ~\leq~ Ke^{-\lambda^2/2\sigma^2}$.
	Then
	\[
		\E\max_{1\leq i\leq N} |V_i| ~\leq~ \sigma\left(\sqrt{2\log(KN)} +
		\frac{1}{\sqrt{2\log(KN)}}\right).
	\]
\end{prop}

Applying Proposition~\ref{prop:Emaxsubg} along with \eqref{eq:Ntautail} allows us to bound the first term in \eqref{eq:telescope for N}:
\begin{align*}
	\E_{n_k}\max_{p_0\in T_0}|N(p_0)| &\leq \frac{\sigma_nM_4}{\sqrt{2}}\left(\sqrt{2\log(\TO + 1)} +
		\frac{1}{\sqrt{2\log(\TO)}}\right) \\
	&\leq \frac{\sigma_nM_4}{\sqrt{2}}\left(\sqrt{2\log(\TO)} + \frac{\TO + 1}{\TO\sqrt{2\log(\TO)}}\right) \\
	&\leq \frac{\sigma_nM_4}{\sqrt{2}}\left(\sqrt{2\log(\TO)} + 0.9\right)\\
	&\leq \sigma_nM_4\left(\sqrt{\log(\TO)}+0.64\right).
\end{align*}
The first line uses the fact that $\TO \le \#T_0\le\TO+1$, the second line follows from the convenient fact that for any $a>1$ we have
\begin{equation}
\label{eq:conv fact 5}
	\sqrt{\log (a+1)}-\sqrt{\log a}\le \left(2a\sqrt{\log a}\right)^{-1},
\end{equation}
and the third line follows from the assumption that $\TO\ge3$.

For the multiscale sum of expected supremums in \eqref{eq:telescope for N}, we can apply Proposition \ref{prop:Emaxsubg} with
$
	\sigma =\frac{1}{2\sqrt{2}}\sigma_nM_4|\Omega||q_j-p_j|= \frac{2^{-j}}{8\sqrt{2}}\sigma_nM_4
$
for every $j\ge0$, and obtain
\begin{align*}
	\sum_{j\geq 0}\E_{n_k}\max_{(p_j,q_j)\in L_j} |N(p_j)-N(q_j)| &\leq
	\frac{\sigma_nM_4}{8\sqrt{2}}\sum_{j\geq 0} 2^{-j}\left(\sqrt{2\log(2^{j+1}\TO + 1)} + \frac{1}{\sqrt{2\log(2^{j+1}\TO)}}\right) \\
	&\leq \frac{\sigma_nM_4}{8\sqrt{2}}\sum_{j\geq 0} 2^{-j}\left(\sqrt{2\log(2^{j+1}\TO)} + \frac{2^{j+1}\TO+1}{2^{j+1}\TO\sqrt{2\log(2^{j+1}\TO)}}\right) \\
	&\leq \frac{\sigma_nM_4}{8\sqrt{2}}\sum_{j\geq 0} 2^{-j}\left(\sqrt{2\log(\TO)} +\sqrt{2(j+1)\log 2}+ 0.62\right)\\
	&\leq \frac{\sigma_nM_4}{8\sqrt{2}}\left(2\sqrt{2\log(\TO)} + 4.42\right)\\
	&\leq \sigma_nM_4\left(0.25\sqrt{\log\left(\TO\right)}+0.4\right),
\end{align*}
where the first line uses the fact that $2^{j+1}\TO\le\#L_j\le2^{j+1}\TO+1$, the second line uses \eqref{eq:conv fact 5}, and the third line uses the assumption that $\TO\ge3$. Consequently, in light of \eqref{eq:telescope for N}, we obtain
\begin{align}
	\E_{n_k}\sup_{\tau\in T}|N(\tau)| &\leq~
	\sigma_nM_4\left(1.25\sqrt{\log(\TO)}+1.04\right).
	\label{eq:ncondbound}
\end{align}

Removing the conditioning on $\{\omega_{k}\}$, we arrive at
\begin{align*}
	\E\sup_{\tau\in T}|N(\tau)| & = \E_{\omega_{k}}\E_{n_{k}}\sup_{\tau\in T}|N(\tau)|\\
	&\leq \sigma_n\left(1.25\sqrt{\log(\TO)}+1.04\right)\E_{\omega_k}\bigmfour \\
 	&\leq \sigma_n\left(1.25\sqrt{\log(\TO)}+1.04\right)\sqrt{\sum_{k=1}^{m}\E_{\omega_k}|\widehat{s}(\omega_{k})|^{2}}\\
 	&= \sigma_n\left(1.25\sqrt{\log(\TO)}+1.04\right)\sqrt{m|\Omega|^{-1}\int_{\Omega}|\widehat{s}(\omega)|^{2}\; d\omega}\\
 	&= \sigma_nM_2\left(1.25\sqrt{\log(\TO)}+1.04\right)\\
	&\leq 2.25\sigma_nM_2\sqrt{\log(\TO)},
\end{align*}
where the third line uses Jensen's inequality, and the last line follows from our assumption that $\TO\ge3$. This completes the proof of Theorem~\ref{thm:noiseexp}.

\subsubsection{Tail Bound (Proof of Theorem~\ref{thm:mainNoisenew})}

\label{sec:noisetail}

Recall that, conditioned on $\{\omega_{k}\}$, $N(\tau)$ is a centered complex-valued
Gaussian process. The following result, proved in Appendix \ref{sec:Proof-of-Gaussian Tail bound},
provides a sharp tail bound for the supremum of this random process.
\begin{lem}
\label{lem:Gaussian tail bound}Let $\{G(t),\, t\in\Delta\}$ be a
centered complex-valued Gaussian process. Define the weak variance
as $\nu^{2}:=\sup_{t\in\Delta}\mathbb{E}\left|G(t)\right|^{2}$.
 Then the following holds for any $\lambda\ge0$: \[
\P{ \sup_{t\in\Delta}\left|G\left(t\right)\right|\ge\mathbb{E}\sup_{t\in\Delta}\left|G\left(t\right)\right|+\lambda} \le\exp\left(-\frac{\lambda^{2}}{2\nu^{2}}\right).\]
 \end{lem}

To apply this bound to $N(\tau)$ conditioned on $\{\omega_{k}\}$,
notice that \eqref{eq:2nd moment of |N(t)|} directly implies $\nu^{2}=\sigma_{n}^{2}M_{4}^{2}$.
Therefore, using Lemma~\ref{lem:Gaussian tail bound}, we obtain the following for
any $\lambda\ge0$:
\begin{equation}
	\PN{\sup_{\tau\in T}\left|N\left(\tau\right)\right|>\E_{n_k}\sup_{\tau\in T}\left|	N\left(\tau\right)\right| +\lambda}
	\le\exp\left(-\frac{\lambda^{2}}{2\sigma_{n}^{2}M_{4}^{2}}\right).
\label{eq:lamtailexp}
\end{equation}

Now, recall \eqref{eq:ncondbound} and note that due to our assumption that $\TO\ge3$, we have $\E_{n_k}\sup_{\tau\in T}|N(\tau)| \leq 2.25 \sigma_nM_4\sqrt{\log(\TO)}$. Combining this fact with \eqref{eq:lamtailexp}, take $\lambda=\sigma_{n}M_{4}\sqrt{2\log\left(2/\delta\right)}$ to get
\begin{align}
	&\PN{ \sup_{\tau\in T}\left|N\left(\tau\right)\right|> 2\sigma_nM_4\max\left(2.25\sqrt{\log\left(\TO\right)},\sqrt{2\log(2/\delta)}\right)}\nonumber\\
	&\le \PN{ \sup_{\tau\in T}\left|N\left(\tau\right)\right|> 2.25\sigma_nM_4\sqrt{\log\left(\TO\right)}+\sigma_nM_4\sqrt{2\log(2/\delta)})}\nonumber\\
	&\le \delta/2.
	\label{eq:noise tail bound cond on omega}
\end{align}

Now, we need only to show that $M_{4}$ is small with high probability.
For this, we will use the Bernstein inequality (Lemma \ref{lem:Bernstein ineq}).
To apply the Bernstein inequality, take $V_{k}=|\widehat{s}(\omega_{k})|^{2}-|\Omega|^{-1}\left\Vert \widehat{s}\right\Vert _{2}^{2}$.
Notice that
$$
	\mathbb{E}V_{k}=\mathbb{E}|\widehat{s}(\omega_{k})|^{2}-|\Omega|^{-1}\left\Vert \widehat{s}\right\Vert _{2}^{2}
	=\int_{\Omega}|\widehat{s}(\omega)|^{2}\,|\Omega|^{-1}d\omega-|\Omega|^{-1}\left\Vert \widehat{s}\right\Vert _{2}^{2}
	=|\Omega|^{-1}\left\Vert \widehat{s}\right\Vert _{2}^{2}-|\Omega|^{-1}\left\Vert \widehat{s}\right\Vert _{2}^{2}
	=0
$$
and also that
\begin{align*}
	|V_{k}| & \leq\sup_{\omega}\left||\widehat{s}(\omega)|^{2}-|\Omega|^{-1}\left\Vert \widehat{s}\right\Vert _{2}^{2}\right|\\
	& \le\sup_{\omega}~\max\left(|\widehat{s}(\omega)|^{2},|\Omega|^{-1}\left\Vert \widehat{s}\right\Vert _{2}^{2}\right)\\
	& =\max\left(\sup_{\omega}|\widehat{s}(\omega)|^{2},|\Omega|^{-1}\left\Vert \widehat{s}\right\Vert _{2}^{2}\right)\\
	& =\max\left(M,|\Omega|^{-1}\left\Vert \widehat{s}\right\Vert _{2}^{2}\right)\\
	& =M,\end{align*}
where the second line uses the convenient fact that $\left|a-b\right|\le\max\left(a,b\right)$ for any $a,b\ge0$, and the fifth line uses the fact that $\left\Vert \widehat{s}\right\Vert _{2}^{2} \le |\Omega|\sup_\omega|\widehat{s}(\omega)|^2$. In addition, we know that
$$
	\mathbb{E}V_{k}^{2} =\mathbb{E}|\widehat{s}(\omega_{k})|^{4}-|\Omega|^{-2}\left\Vert \widehat{s}\right\Vert _{2}^{4}
	 \le\mathbb{E}|\widehat{s}(\omega_{k})|^{4}
	 =|\Omega|^{-1}\|\widehat{s}\|_{4}^{4},
$$
 where the first equality follows from the convenient fact that $\mathbb{E}|V-\mathbb{\mathbb{E}}V|^{2}=\mathbb{E}|V|^{2}-|\mathbb{E}V|^{2}$ for any random variable $V$. Therefore, we have $\rho^{2}=\sum_{k=1}^{m}\mathbb{E}V_{k}^{2}\le m|\Omega|^{-1}\|\widehat{s}\|_{4}^{4}=M_{1}^{2}$.
Now we can apply the Bernstein inequality to $\sum_{k=1}^{m}V_{k}$
for any $\lambda\ge0$ and obtain
\[
	\P{ M_{4}^{2}>m|\Omega|^{-1}\left\Vert \widehat{s}\right\Vert _{2}^{2}+\lambda} ~\leq~\exp\left(-\frac{\lambda^{2}}{2M_{1}^{2}+2M\lambda/3}\right).
\]

Suppose first that $M_{1}\ge(M/3)\sqrt{\log(2/\delta)}$.
Then take $\lambda=aM_{1}\sqrt{\log(2/\delta)}$ for some $a\ge0$
to get
\begin{align*}
	\P{ M_{4}^{2}\ge m|\Omega|^{-1}\left\Vert \widehat{s}\right\Vert _{2}^{2}+aM_{1}\sqrt{\log(2/\delta)}}  & \le\exp\left(-\frac{a^{2}M_{1}^{2}\log(2/\delta)}{2M_{1}^{2}+2aMM_{1}\sqrt{\log(2/\delta)}/3}\right)\\
	& \le\exp\left(-\frac{a^{2}M_{1}^{2}\log(2/\delta)}{2M_{1}^{2}+2aM_{1}^{2}}\right)\\
	& \le\exp\left(-\frac{a^{2}}{2+2a}\log\left(2/\delta\right)\right)\\
	& \le\delta/2,
\end{align*}
 which is valid for $a\ge1+\sqrt{3}$.

Now suppose that $M_{1}<(M/3)\sqrt{\log(2/\delta)}$. Take $\lambda=aM\log(2/\delta)$.
Then
\begin{align*}
	  \P{ M_{4}^{2}\ge m|\Omega|^{-1}\left\Vert \widehat{s}\right\Vert _{2}^{2}+aM\log(2/\delta)}  & \leq\exp\left(-\frac{a^{2}M^{2}\log^{2}(2/\delta)}{2M_{1}^{2}+(2a/3)M^{2}\log(2/\delta)}\right)\\
	  & <\exp\left(-\frac{a^{2}M^{2}\log^{2}(2/\delta)}{(2/9)M^{2}\log(2/\delta)+(2a/3)M^{2}\log(2/\delta)}\right)\\
	  & =\exp\left(-\frac{a^{2}\log(2/\delta)}{2/9+2a/3}\right)\\
	  & \leq\delta/2,
\end{align*}
 which is valid for $a\ge\frac{1+\sqrt{3}}{3}$.

So with probability at least $1-\delta/2$ we will have
\begin{equation}
	M_{4}^{2}~\leq~m|\Omega|^{-1}\left\Vert \widehat{s}\right\Vert _{2}^{2}+a\max\left(M_{1}\sqrt{\log(2/\delta)},M\log(2/\delta)\right),\label{eq:nboundtwoterms}
\end{equation}
 for $a\ge1+\sqrt{3}$.

By the definition of $M_{1}$, note that $m|\Omega|^{-1}\left\Vert \widehat{s}\right\Vert _{2}^{2}\ge aM_{1}\sqrt{\log(2/\delta)}$
is equivalent to $m|\Omega|^{-1}\left\Vert \widehat{s}\right\Vert _{2}^{2}\ge a\sqrt{m}\sqrt{\left|\Omega\right|^{-1}}\left\Vert \widehat{s}\right\Vert _{4}^{2}\sqrt{\log(2/\delta)}$,
which is equivalent to $m\geq a^{2}\frac{\|\widehat{s}\|_{4}^{4}}{\|\widehat{s}\|_{2}^{4}}|\Omega|\log(2/\delta)$.
Also, $m|\Omega|^{-1}\left\Vert \widehat{s}\right\Vert _{2}^{2}\ge aM\log(2/\delta)$
is equivalent to $m\ge a\frac{M}{\|\widehat{s}\|_{2}^{2}}|\Omega|\log(2/\delta)$.
Therefore, in order for the first term on the right hand side of (\ref{eq:nboundtwoterms})
to be dominant we conveniently assume that
\begin{equation}
	m\geq a^{2}\max\left(\frac{|\Omega|\|\widehat{s}\|_{4}^{4}}{\|\widehat{s}\|_{2}^{4}},\frac{|\Omega|M}{\|\widehat{s}\|_{2}^{2}}\right)\log(2/\delta),\label{eq:conv cond on m}
\end{equation}
 where we used the fact that $a>1$. Under this assumption,
\begin{equation}
	M_{4}^{2}\leq2m|\Omega|^{-1}\left\Vert \widehat{s}\right\Vert _{2}^{2}=2\bigmtwo^{2}\label{eq:bound on M4}
\end{equation}
 with probability exceeding $1-\delta/2$. Let $\mathcal{E}$ denote
the event specified in \eqref{eq:bound on M4}; clearly if \eqref{eq:conv cond on m} is met, then $\P{ \mathcal{E}} \ge1-\delta/2$. Now, using \eqref{eq:noise tail bound cond on omega}, we have
\begin{align}
	& \P{ \sup_{\tau\in T}\left|N\left(\tau\right)\right|>2\sqrt{2}\sigma_nM_2\max\left(2.25\sqrt{\log\left(\TO\right)},\sqrt{2\log(2/\delta)}\right)\,\middle|\,\mathcal{E}} \nonumber \\
	& \le\P{ \sup_{\tau\in T}\left|N\left(\tau\right)\right|>2\sigma_nM_4\max\left(2.25\sqrt{\log\left(\TO\right)},\sqrt{2\log(2/\delta)}\right)\,\middle|\,\mathcal{E}} \nonumber \\
	& = \frac{1}{\P{ \mathcal{E}}}\int_{\mathcal{E}}\PN{ \sup_{\tau\in T}\left|N\left(\tau\right)\right|>2\sigma_nM_4\max\left(2.25\sqrt{\log\left(\TO\right)},\sqrt{2\log(2/\delta)}\right)\,\middle|\,\{\omega_{k}\}} \, d\mu(\{\omega_{k}\}) \nonumber \\
	& \le\frac{\delta}{2\P{ \mathcal{E}}}\int_{\mathcal{E}}d\mu(\{\omega_{k}\})\nonumber \\
	& =\frac{\delta}{2},\label{eq:noise tail bound cond on e}
\end{align}
where the third line above follows from the definition of conditional measure.

On the other hand, for any $u\ge0$, we have
\begin{align*}
	\P{ \sup_{\tau\in T}\left|N\left(\tau\right)\right|\ge u}  & =\P{ \sup_{\tau\in T}\left|N\left(\tau\right)\right|\ge u\,\middle|\,\mathcal{E}} \P{ \mathcal{E}} +\P{ \sup_{\tau\in T}\left|N\left(\tau\right)\right|\ge u\,\middle|\,\mathcal{E}^{C}} \P{ \mathcal{E}^{C}} \\
	& \le\P{ \sup_{\tau\in T}\left|N\left(\tau\right)\right|\ge u\,\middle|\,\mathcal{E}} +\frac{\delta}{2}.
\end{align*}
 Taking $u=2\sqrt{2}\sigma_nM_2\max\left(2.25\sqrt{\log\left(\TO\right)},\sqrt{2\log(2/\delta)}\right)$
and using \eqref{eq:noise tail bound cond on e}, we finally obtain
\begin{align*}
	& \P{ \sup_{\tau\in T}\left|N\left(\tau\right)\right|>2\sqrt{2}\sigma_nM_2\max\left(2.25\sqrt{\log\left(\TO\right)},\sqrt{2\log(2/\delta)}\right)} \\
	& \le\P{ \sup_{\tau\in T}\left|N\left(\tau\right)\right|>2\sqrt{2}\sigma_nM_2\max\left(2.25\sqrt{\log\left(\TO\right)},\sqrt{2\log(2/\delta)}\right)\,\middle|\mathcal{E}} +\frac{\delta}{2}\\
	& \le \delta,
\end{align*}
 which is valid when \eqref{eq:conv cond on m}
is met and $a\ge1+\sqrt{3}$. Setting $a=1+\sqrt{3}$, $C_3 = a^2$, and $C_4 = \frac{2.25 \sqrt{2}}{\pi}$, we complete the proof
of Theorem~\ref{thm:mainNoisenew}.

\appendix

\section{Proof of Lemma \ref{lem:Y to Z}\label{sec:Y to Z}}

Let $\E_{\omega_{k}}$ and $\E_{\omega'_{k}}$ denote expectation
with respect to $\{\omega_{k}\}$ and $\{\omega_{k}'\}$, respectively.
Then we can write
 \begin{align*}
\E\sup_{\tau}|Y(\tau)| & =\E_{\omega_{k}}\sup_{\tau}\left|Y(\tau)-\E_{\omega'_{k}}Y'(\tau)\right|\quad & \text{(\ensuremath{Y'(\tau)}\,\ is zero mean)}\\
 & =\E_{\omega_{k}}\sup_{\tau}\left|\E_{\omega'_{k}}\left[Y(\tau)-Y'(\tau)\right]\right|\quad & \text{(independence; \ensuremath{\E_{\omega'_{k}}Y=Y})}\\
 & \leq\E_{\omega_{k}}\E_{\omega'_{k}}\sup_{\tau}\left|Y(\tau)-Y'(\tau)\right|\quad & \text{(Jensen's inequality, sup\ensuremath{\left|\cdot\right|}\,\ is convex)}\\
 & =\E\sup_{\tau}\left|Y(\tau)-Y'(\tau)\right|\quad & \text{(iterated expectation)}\\
 & =\E\sup_{\tau}|Z(\tau)|.
\end{align*}
Finally, since $Z(\tau)$ has the same distribution as $Z'(\tau)$, $\E\sup_{\tau}|Z(\tau)| = \E\sup_{\tau}|Z'(\tau)|$.

\section{Proof of Lemma \ref{lem:Y to Z2}\label{sec:Y to Z2}}

Recall that, for every $\tau\in T$, $Y(\tau)$ was symmetrized by creating an independent copy $Y'(\tau)$ and defining $Z(\tau)=Y(\tau)-Y'(\tau)$.
Therefore, for any $a,\lambda>0$, the occurrence of the events $\left\{\sup_{\tau}|Y(\tau)|>a+\lambda\right\}$ and $\left\{\sup_{\tau}|Y'_\tau|< a\right\}$ imply that
$$
	\lambda <\sup_{\tau}|Y(\tau)|-\sup_{\tau}|Y'(\tau)|
	 \le \sup_{\tau}\left( |Y(\tau)|-|Y'(\tau)| \right)
	 \le \sup_{\tau} |Y(\tau)-Y'(\tau)|
	 =\sup_{\tau} |Z(\tau)|.
$$
Setting $a=2\E\sup_{\tau}|Y(\tau)|$ and recalling that $Z(\tau)$ and $Z'(\tau)$ have the same distribution, we conclude that
\begin{align*}
\P{\sup_{\tau}|Z'(\tau)|>\lambda} &\ge \P{ \left\{\sup_{\tau}|Y(\tau)|>a+\lambda\right\} \cap \left\{\sup_{\tau}|Y'(\tau)|< a\right\}} \\
& = \P{\sup_{\tau}|Y(\tau)|>2\E\sup_{\tau}|Y(\tau)|+\lambda} \, \P{\sup_{\tau}|Y(\tau)|< 2\E\sup_{\tau}|Y(\tau)|}\\
& \ge \frac{1}{2} \P{\sup_{\tau}|Y(\tau)|>2\E\sup_{\tau}|Y(\tau)|+\lambda},
\end{align*}
where the second line follows from the fact that $Y$ and $Y'$ are independent copies of the same random process, and the third line follows from applying the Markov inequality, which states that for any nonnegative random variable $V$ and any $c>0$, we have $\P{V\ge c}\le\ c^{-1}\E V$.

\section{Proof of Lemma \ref{lem:inctails}}
\label{sec:inctails}

Hoeffding's inequality for Rademacher sums~\cite[Prop.~6.11]{ra09-1} states that if $b_1,b_2,\ldots,b_N$ are complex numbers and $\epsilon_1,\epsilon_2,\ldots,\epsilon_N$ is a Rademacher series, then for every $\lambda\ge0$ we have
\begin{equation}
	\label{eq:complexhding}
	\P{\left|\sum_{k=1}^N \epsilon_kb_k\right| > \lambda} ~\leq~ 2\exp\left(\frac{-\lambda^2}{2\sum_k|b_k|^2}\right).
\end{equation}
Conditioned on $\{\omega_k\}$ and $\{\omega'_k\}$, we can write
\[
	Z'(\tau) = \sum_{k=1}^m\epsilon_k\alpha_k, \quad\text{for some $\alpha_k$ with}~~~
	|\alpha_k| =
	\left|A\left|\widehat{s}(\omega_{k})\right|^{2}e^{\iunit\omega_{k}(\tau-\tau_{0})}-
	A\left|\widehat{s}(\omega_{k}')\right|^{2}e^{\iunit\omega_{k}'(\tau-\tau_{0})}\right|,
\]
and therefore,
\begin{align*}
	\sum_{k=1}^m|\alpha_k|^2 &\leq |A|^2\sum_{k=1}^m \left|
	\left|\widehat{s}(\omega_{k})\right|^{2}e^{\iunit\omega_{k}(\tau-\tau_{0})}-
	\left|\widehat{s}(\omega_{k}')\right|^{2}e^{\iunit\omega_{k}'(\tau-\tau_{0})}
	\right|^2 \\
	&\leq
	|A|^2\sum_{k=1}^{m}2\left|\widehat{s}(\omega_{k})\right|^{4}+2\left|\widehat{s}(\omega_{k}')\right|^{4} \\
	&= 2|A|^2M_3^2,
\end{align*}
where the second line uses (\ref{eq:conv fact 2}). Plugging into \eqref{eq:complexhding} yields \eqref{eq:tailinc1} as desired.

For the increment bound, conditioned on $\{\omega_k\}$ and $\{\omega'_k\}$, we can write
\[
	Z'(\tau_1) - Z'(\tau_2) = \sum_{k=1}^m\epsilon_k\beta_k,
\]
where
\begin{align*}
	\sum_{k=1}^m|\beta_k|^2  &= |A|^2 \sum_{k=1}^{m}\Biggl|\left|\widehat{s}(\omega_{k})\right|^{2}e^{\iunit\omega_{k}(\tau_{1}-\tau_{0})}-\left|\widehat{s}(\omega'_{k})\right|^{2}e^{\iunit\omega'_{k}(\tau_{1}-\tau_{0})}-\left|\widehat{s}(\omega_{k})\right|^{2}e^{\iunit\omega_{k}(\tau_{2}-\tau_{0})}+\left|\widehat{s}(\omega'_{k})\right|^{2}e^{\iunit\omega'_{k}(\tau_{2}-\tau_{0})}\Biggl|^{2} \\
&\leq~
	8|A|^2\sum_{k=1}^{m}\biggl(\left|\widehat{s}(\omega_{k})\right|^{4}\left|\sin\left(\frac{1}{2}\omega_{k}(\tau_{1}-\tau_{2})\right)\right|^{2} +  \left|\widehat{s}(\omega_{k}')\right|^{4}\left|\sin\left(\frac{1}{2}\omega'_{k}(\tau_{2}-\tau_{1})\right)\right|^{2}\biggl) \\
&\leq  8|A|^2\sum_{k=1}^{m}\left|\widehat{s}(\omega_{k})\right|^{4}\left|\frac{1}{2}\omega_{k}(\tau_{1}-\tau_{2})\right|^{2}+\left|\widehat{s}(\omega_{k}')\right|^{4}\left|\frac{1}{2}\omega'_{k}(\tau_{2}-\tau_{1})\right|^{2} \\
 & \leq  8|A|^2\sum_{k=1}^{m}\left|\widehat{s}(\omega_{k})\right|^{4}\left|\frac{|\Omega|}{4}(\tau_{1}-\tau_{2})\right|^{2}+\left|\widehat{s}(\omega_{k}')\right|^{4}\left|\frac{|\Omega|}{4}(\tau_{2}-\tau_{1})\right|^{2} \\
 & \leq  \frac{1}{2}|A|^2\bigmthree^2|\Omega|^2|\tau_{1}-\tau_{2}|^2.
\end{align*}
The second line above follows from applying \eqref{eq:conv fact 2} and then \eqref{eq:conv fact 1}, the third line uses the fact that $|\sin(a)|\le|a|$, and the fourth line follows from the fact that $\Omega$ is symmetric about the origin, i.e., that $\Omega=[-\omega_{\mathrm{max}},\omega_{\mathrm{max}}]$. Plugging into \eqref{eq:complexhding} yields \eqref{eq:tailinc2} as desired.

\section{Proof of Proposition~\ref{prop:Emaxsubg}}
\label{sec:Proof of Emaxsubg}

It follows from \eqref{eq:tailtoE} that
\[
	\E\max_i|V_i|=\int_{0}^\infty \P{\max_i|V_i|>\lambda}~d\lambda.
\]
By breaking the integration interval at $\lambda_0:=\sigma \sqrt{2\log(KN)}$, we can write
\begin{align*}
	 \E \max_{i} |V_i| & = \int_0^{\lambda_0} \P{\max_i|V_i|>\lambda}~d\lambda + \int_{\lambda_0}^\infty \P{\max_i|V_i|>\lambda}~d\lambda\\
	  &\le \int_0^{\lambda_0} d\lambda + \int_{\lambda_0}^\infty N\max_i \P{|V_i|>\lambda}~d\lambda\\
	  &\le \lambda_0+\int_{\lambda_0}^\infty KN e^{-\lambda^2/2\sigma^2}~d\lambda\\
	  &\le \lambda_0 +KN\cdot\frac{\sigma^2}{\lambda_0}e^{-\lambda_0^2/2\sigma^2}\\
	  &= \sigma \left(\sqrt{2\log(KN)} +\frac{1}{\sqrt{2\log(KN)}}\right),
\end{align*}
as claimed. The second line above uses the union bound, and the fourth line uses \eqref{eq:conv fact 3}.

\section{Proof of Lemma \ref{lem:Gaussian tail bound} \label{sec:Proof-of-Gaussian Tail bound}}

The proof essentially follows~\cite[p. 134]{Phenomenon-Ledoux}.
Fix $t_{1},t_{2},...,t_{N}$ in $\Delta$ and form the vector $g:=[G(t_{1}),G(t_{2}),...,G(t_{N})]^{T}\in\mathbb{C}^{N}$
with covariance matrix $\Gamma=\Xi\Xi^{*}$. This vector has the same
distribution as $\Xi h$, where $h\in\mathbb{R}^{N}$ is the standard
Gaussian vector, whose entries are i.i.d.\ zero-mean Gaussian random
variables with unit variance. This is because $\mathbb{E}\Xi h=\Xi\cdot\mathbb{E}h=0_{N}=\mathbb{E}g$
and $\mathbb{E}\left\{ \left(\Xi h\right)\left(\Xi h\right)^{*}\right\} =\Xi\mathbb{E}\left\{ hh^{T}\right\} \Xi^{*}=\Xi\Xi^{*}=\Gamma=\mathbb{E}\left\{ gg^{*}\right\} $.
Here $0_{N}$ denotes the $N\times1$ zero vector. Now consider the
function $F(x):\,\mathbb{R}^{N}\rightarrow\mathbb{R}$ defined as
$F(x):=\max_{1\le i\le N}\left|\left(\Xi x\right)_{i}\right|=\left\Vert \Xi x\right\Vert _{\infty}$.
Let $x_{1},\, x_{2}\in\mathbb{R}^{N}$. Now we can write \begin{align*}
\left|F(x_{1})-F(x_{2})\right| & =\left|\left\Vert \Xi x_{1}\right\Vert _{\infty}-\left\Vert \Xi x_{2}\right\Vert _{\infty}\right|\\
 & \le\left\Vert \Xi\left(x_{1}-x_{2}\right)\right\Vert _{\infty}\\
 & \le\left\Vert \Xi\right\Vert _{\infty,2}\left\Vert x_{1}-x_{2}\right\Vert _{2}.\end{align*}
 The second line follows from the triangle inequality, and in the third line, $\left\Vert \Xi\right\Vert _{\infty,2}$
denotes the operator norm of $\Xi$ from $\mathbb{R}^{N}$ equipped with
the $l_{2}$-norm to $\mathbb{C}^{N}$ equipped with the $l_{\infty}$-norm.
Let $B_{2}^{N}$ denote the unit $l_{2}$-ball in $\mathbb{R}^{N}$.
Note that we have \begin{align*}
\left\Vert \Xi\right\Vert _{\infty,2}^{2} & = \left(\sup_{x\in B_{2}^{N}}\max_{1\le i\le N}\left|\left(\Xi x\right)_{i}\right| \right)^{2}\\
 & =\sup_{x\in B_{2}^{N}}\max_{1\le i\le N}\left|\left(\Xi x\right)_{i}\right|^{2}\\
 & =\max_{1\le i\le N}\sup_{x\in B_{2}^{N}}\left|\sum_{j=1}^{N}\Xi_{i,j}x_{j}\right|^{2}\\
 & =\max_{1\le i\le N}\,\sum_{j=1}^{N}\left|\Xi_{i,j}\right|^{2}\\
 & =\max_{i}\,\left(\Xi\Xi^{*}\right)_{i,i}\\
 & =\max_{i}\,\Gamma_{i,i}\\
 & =\max_{i}\,\mathbb{\mathbb{E}}\left|G(t_{i})\right|^{2}\\
 & =:\nu_{N}^{2}.\end{align*}
 The third line is a consequence of the Cauchy-Schwartz inequality
for complex-valued numbers. Therefore, $F(\cdot)$ is a $\nu_{N}$-Lipschitz
function of its argument (which is a standard Gaussian vector here).
So we can invoke, for example,~\cite[eq. (2.35)]{Phenomenon-Ledoux}
to get\begin{align*}
\Pr\left\{ F\left(h\right)\ge\mathbb{E}F(h)+\lambda\right\}  & =\Pr\left\{ \max_{1\le i\le N}\left|\left(\Xi h\right)_{i}\right|\ge\mathbb{E}\max_{1\le i\le N}\left|\left(\Xi h\right)_{i}\right|+\lambda\right\} \\
 & =\Pr\left\{ \max_{1\le i\le N}\left|G(t_{i})\right|\ge\mathbb{E}\max_{1\le i\le N}\left|G(t_{i})\right|+\lambda\right\} \\
 & \le\exp\left(-\frac{\lambda^{2}}{2\nu_{N}^{2}}\right).\end{align*}
 Now we can apply monotone convergence to the above inequality; as
$N\rightarrow\infty$, we have $\nu_{N}^{2}\rightarrow\nu^{2}$ and
we come to the desired result.

\bibliographystyle{plain}
\bibliography{References}

\end{document}